\documentclass[lettersize,journal]{IEEEtran} 
\usepackage{amssymb}
\usepackage{changepage}
\usepackage{cite}
\usepackage{amsmath,amssymb,amsfonts}
\usepackage{algpseudocode}
\usepackage{textcomp}
\usepackage{xcolor}
\usepackage{amsmath}
\usepackage{tabularx}
\usepackage{amsthm}
\usepackage{multirow}
\usepackage{relsize}
\usepackage{optidef}
\usepackage{bbold}
\usepackage{graphicx,subcaption,caption}
\usepackage{import}
\usepackage{algorithm2e}
\RestyleAlgo{ruled}

\def\ls{\mspace{4mu}} 
\def\qm{\boldsymbol{q}_m}
\def\ql{\boldsymbol{q}_\ell}

\def\gml{{g}_{m,\ell}}
\def\Gm{\boldsymbol{{G}}}

\def\Gmh{\boldsymbol{{\hat{G}}}}
\def\Gmt{\boldsymbol{{\tilde{G}}}}

\def\Tmat{\boldsymbol{{T}}}
\def\Pmat{\boldsymbol{{P}}}

\def\x{\boldsymbol{x}}

\def\hgml{\hat{g}_{m,\ell}}

\def\yb{\boldsymbol{y}}

\makeatletter
\newcommand{\removelatexerror}{\let\@latex@error\@gobble}
\makeatother
\newtheorem{theorem}{Theorem}

\begin{document}
\title{A  Reinforcement Learning Approach for Wildfire  Tracking with UAV Swarms}

\author{Carles Diaz-Vilor,$^{1}$ Angel Lozano,$^{2}$ and Hamid Jafarkhani$^{1}$
\thanks{$^{1}$C. Diaz-Vilor and H. Jafarkhani are with the Center for Pervasive Communications and Computing, Univ. of California, Irvine. Their work was supported in part by NSF Award CNS-2209695. Emails: 
        {\tt\small \{cdiazvil, hamidj\}@uci.edu }} %
\thanks{$^{2}$A. Lozano is with Universitat Pompeu Fabra (UPF), Barcelona. His work was supported by MICIU/AEI /10.13039/501100011033 under the Maria de Maeztu Units of Excellence Programme (CEX2021-001195-M), by ICREA, and by the Fractus-UPF Chair on 6G. Email:
        {\tt\small  angel.lozano@upf.edu }}
}

\maketitle
\begin{abstract}
Suitably equipped with cameras and sensors, uncrewed aerial vehicles (UAVs) can be instrumental for wildfire prediction, tracking, and monitoring, provided that
uninterrupted connectivity can be guaranteed even if
some of the ground access points (APs) are damaged by the fire itself. A cell-free network structure, with UAVs connecting to a multiplicity of APs, 
is therefore ideal in terms of resilience.  
This work proposes a trajectory optimization framework for a UAV swarm tracking a wildfire while maintaining cell-free connectivity with ground APs. Such optimization entails a constant repositioning of the multiplicity of UAVs as the fire evolves to ensure that the best possible view is acquired and transmitted reliably, while respecting altitude limits, avoiding collisions, and proceeding to recharge batteries as needed.
Given the complexity and time-varying nature of this multi-UAV trajectory optimization, reinforcement learning is leveraged, specifically the twin-delayed deep deterministic policy gradient algorithm. The approach is shown to be highly effective for wildfire tracking and coverage and could be likewise applicable to survey other natural and man-made phenomena, including weather events, earthquakes, or chemical spills.
\end{abstract}


\IEEEpeerreviewmaketitle

\section{Introduction}

Wildfires are a growing concern worldwide and
cause major 
environmental and economic damage. It is therefore crucial to predict, track, and monitor them to facilitate the actions of firefighters and emergency responders  \cite{abatzoglou2016impact,chuvieco2010development,THOMPSON20111895}.
Additionally, 
the study of wildfires can help to identify areas at risk and to understand the nature of fire propagation, along with possible causes and factors that contribute to their spread \cite{cohen2010wildland,finney2011simulation,FireCit,Finney_1998}.

An attractive means to monitor, track, and sense events is deploying uncrewed aerial vehicles (UAVs) given the (\emph{i}) availability of cameras and sensors, (\emph{ii}) wireless connectivity to ground access points (APs), and (\emph{iii}) easy control and maneuverability \cite{CTN2022,8049328,9739676,8119562,9195795}. UAVs can fly at low altitudes and collect data ranging from high-resolution images to heat signatures, subsequently conveying those to a ground network of APs \cite{9147613,7587184,7739740,seraj2022multi,yuan2017fire,yuan2015uav}. 
{We hasten to emphasize that UAVs surveying a wildfire act as end devices in the network, a role that is dual to that of flying base stations in the UAV-assistance paradigm \cite{alzahrani2020uav,oubbati2023multi,oubbati2021dispatch}}.

The transmission of the gathered data is contingent on the connectivity. Regardless of whether APs are damaged by fire, the flow of information should be guaranteed. This relates to the network's resiliency and, in that respect, it is relevant that wireless systems are
evolving from traditional cellular structures towards cell-free arrangements. These are especially appealing when a high degree of reliability is needed, given that users can then communicate with multiple APs \cite{venkatesan2007network,7827017,8097026,8845768,9043895,9247465,9064545,8952782,demir2020joint,9684861}.
Initial results have demonstrated the efficacy and benefits of a cell-free architecture for UAVs \cite{9930941,CFUAVDepl,9453784,9336017,10186347}
and hence this is the structure adopted here.

The design of deployments and trajectories for UAVs in cellular and wireless sensor networks 
is a problem of growing interest \cite{9086619,8519749,Globec,8708979, Zeng2018TrajectoryMulticasting,Wu2018JointNetworks,7509638,8698468,8302930,Cheng2018UAVCells,9199120,Zeng2017Energy-EfficientOptimization,9387137, Zeng2019EnergyUAV,9658259,CommSensUAV,9802837}, yet, to the best of our knowledge, this is the first work that considers a multi-UAV trajectory optimization to track an event with cell-free connectivity.
This gives rise to various challenges,
mainly related to the complexity and time-varying nature of the problem,  which precludes the use of classical optimization techniques. Compounding this complexity, the lifespan of UAV batteries is limited and thus, once a UAV is low on energy, a recharge is needed. Authors in \cite{ucgun2021review} provide an exhaustive review of different techniques to recharge or replace UAV batteries, along with the required time. Altogether, the problem at hand breaks down into two stages of \emph{tracking} and \emph{charging}. Two decidedly nonconvex optimization problems arise, one for each stage, with the switching from tracking to charging based on the remaining energy at each UAV.

Reinforcement learning (RL) provides a framework where an agent learns optimal policies based on interactions with the environment and on feedback in the form of rewards \cite{sutton2018reinforcement}. For settings with a large number of states and/or actions, strategies combining RL with neural networks have been devised, converging to the so-called deep Q-learning (DQL) \cite{8103164,mnih2013playing, sutton1999policy, konda1999actor, mnih2016asynchronous, lillicrap2015continuous}.
These methods have already been considered for UAV trajectory optimizations \cite{9453811,9455139,9171468,9623508,9363308}. A recent actor-critic  algorithm named twin-delayed deep deterministic policy gradient (TD3) has been shown to perform better than its predecessors 
\cite{fujimoto2018addressing} in terms of stability, exploration capabilities, handling of continuous actions, and  sampling efficiency. 
The field of UAV trajectory optimization has greatly benefited from this new algorithm \cite{9426899,9508149,9504602,10086561}, which is applied in this work as well.  The main contributions of the paper are as follows:

\begin{itemize}
    \item An analytical framework is set forth to describe the tracking of a wildfire by a swarm of UAVs equipped with cameras. 
    The connectivity between UAVs and APs is through a cell-free network, with Rician fading, channel estimation, and minimum mean square error (MMSE) reception explicitly modeled.
    \item Based on communication and mechanical constraints, the multi-UAV {tracking} and {charging}  optimization problems are formulated.
    \item For each of the two problems, the UAV trajectory and the transmit power optimizations are confronted using the TD3 approach. While, for the {tracking} stage, the objective is to monitor the wildfire perimeter ensuring a correct reception of the video/images in a timely fashion, during the {charging} stage the goal is to reach a charging point with the minimum energy expenditure.
    \item The impact on the wildfire coverage of parameters such as the number of UAVs, the allowed flying altitudes, and the UAVs' energy   is established. A similar analysis is conducted for the {charging} problem. 
\end{itemize}

The remainder of the manuscript is organized as follows. Sec. \ref{Sec:Model_Fire} presents the camera and communication models. In Sec. \ref{Sec:CF_Fire}, the cell-free connectivity is introduced while 
Sec. \ref{Sec:Problem_Fire} formalizes the two optimization problems. 
Sec. \ref{Sec:Sol_Fire} subsequently focuses on  the solution of those problems while numerical results are  discussed in Sec. \ref{Sec:SimRes_Fire}. Concluding remarks are provided in Sec. \ref{Sec:Concl_Fire}.

\begin{figure*}
    \begin{align}\label{eq:FOV}
    \mathcal{B}_m^{(n)} = \Big \{ \boldsymbol{v} =(v_x,v_y): | x_m^{(n)} - v_x| \leq h_m^{(n)} \tan(\alpha_1) \text{ and } | y_m^{(n)} - v_y| \leq h_m^{(n)} \tan(\alpha_2)\Big \} 
\end{align}
\end{figure*}

\section{System Model}\label{Sec:Model_Fire}

The system under consideration features $M$ UAVs, each equipped with a video camera. The $m$th UAV is located at $\qm^{(n)}= \big( x_{m}^{(n)}, y_m^{(n)}, h_m^{(n)}\big)$, where the height $h_m^{(n)}$ is the distance from the focal point of its camera's lens to the ground,  $n$ is the time index, and the duration of each time slot is $\delta$. 
The UAVs are served by $L$  APs with locations $\ql= ({x}_\ell,{y}_\ell, h_\ell)$. Additionally, the system contains $C$ UAV charging stations located at $\boldsymbol{c}_c = ({\rm x}_c, {\rm y}_c, {\rm h}_c)$.

\subsection{Camera Model}\label{sec:Camera}

Consider a planar environment $\mathcal{F} \in \mathbb{R}^2$ and define the field of view (FoV) as the area projected over $\mathcal{F}$ that a camera captures. For a UAV with a downward-facing camera, a rectangular FoV $\mathcal{B}_m^{(n)}$ is defined in \eqref{eq:FOV} atop the next page,
where $\alpha_1$ and $\alpha_2$ represent the two halfview angles associated with the perpendicular edges of a rectangle (see Fig. \ref{fig:FoV}). The notion of \emph{area per pixel} expresses the tradeoff between the quality of the image and the dimension of the FoV: higher-resolution images correspond to  smaller FoVs, and vice versa. From classical optics, the area per pixel is  \cite{5959179}
\begin{align}\label{eq:cameraarea}
    f\big(\qm^{(n)}, \boldsymbol{v}\big) = \left\{\begin{array}{ll} 
                a \,\left(b - h_m^{(n)} \right)^2 & \boldsymbol{v} \in  \mathcal{B}_m^{(n)}\\
                \infty & \text{otherwise} \\
                \end{array} \right.,
\end{align}
where {$b$ and $a$ depend on the camera capabilities, representing the}  focal length of the lens and the area of a pixel on the lens divided by $b^2$, respectively, while $\boldsymbol{v}$ represents the position on $\mathcal{F}$.
At higher altitudes, $f(\cdot)$ is larger because each pixel covers a bigger area on the FoV, and the resolution is lower. At lower altitudes, it is the other way around. Outside the FoV, by definition, there are no pixels and $f(\cdot)$ is therefore infinite.

The camera's sensor consists of $I$ equispaced pixels.
From sheer geometry, 
\begin{align}
   I = \frac{4 \big( b - h_m^{(n)} \big)^2 \tan (\alpha_1) \tan (\alpha_2)}{f\big(\qm^{(n)}, \boldsymbol{v}\big)}.
\end{align}
For a 24-bit RGB color system, the number of bits to represent a picture  is $24 I$ and, compressing the image with a ratio of $\rho$, the number of bits to be transmitted per image is
\begin{align}\label{eq:n_bits_img}
    B 
    & = 24 I \rho.
\end{align}

\begin{figure}[]
     \centering
     \includegraphics[scale=0.38]{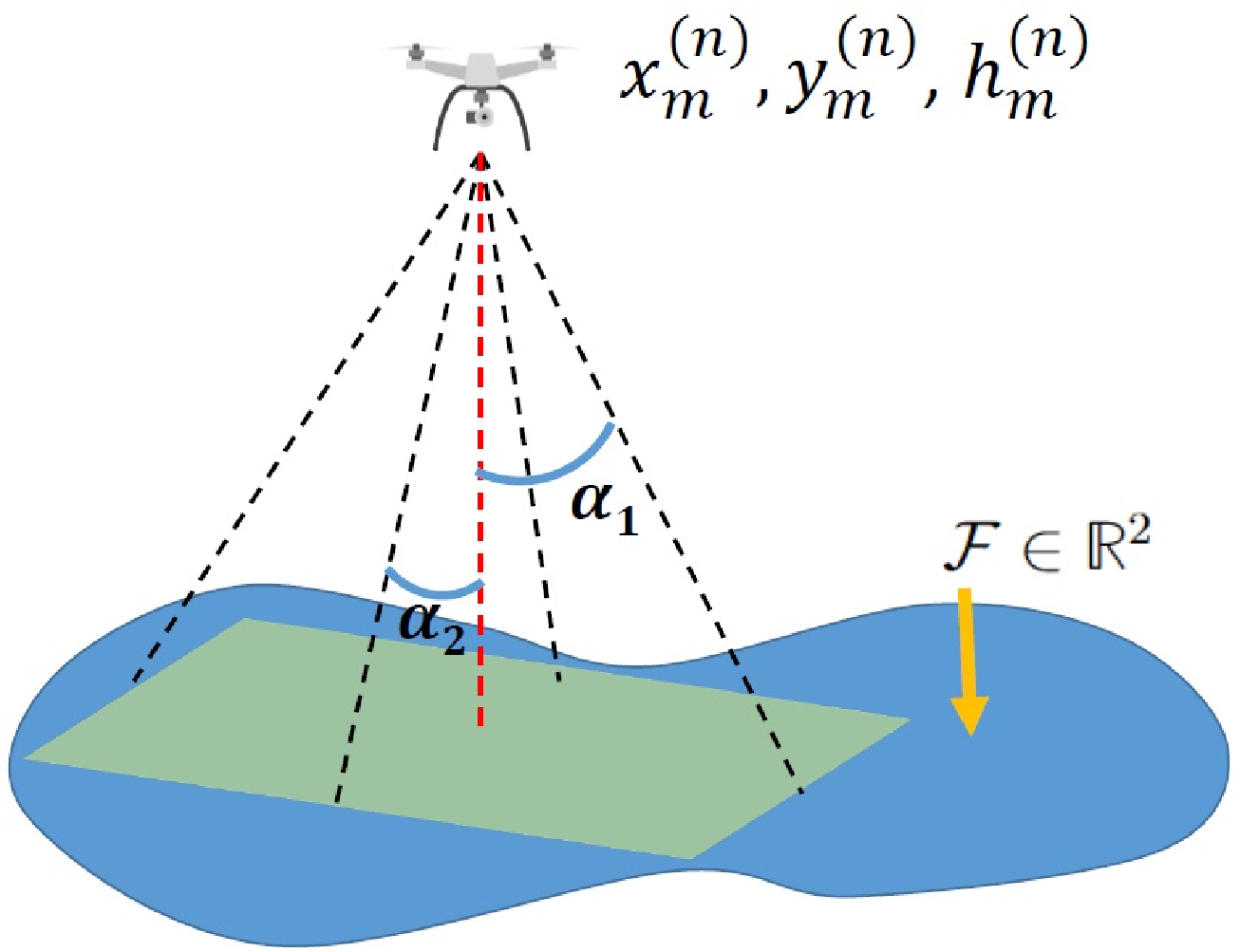}
     \caption{Rectangular FoV of the $m$th UAV for given $\alpha_1$ and $\alpha_2$ over a planar region $\mathcal{F}$.}
     \label{fig:FoV}
\end{figure}

\subsection{Channel Model}

{Upon imaging the environment, UAVs establish wireless connections with the APs to convey the captured images. The air-to-ground } channel coefficient between the $m$th UAV and the $\ell$th AP is denoted by $\gml ^{(n)}$, drawn from a Rician distribution such that \cite[Sec. 3.4.1]{heath2018foundations}
\begin{align}\label{eq:channel}
    g_{m,\ell}^{(n)}  = \sqrt{\frac{\beta_0 }{\big(d^{(n)}_{m,\ell}\big)^\kappa ( K_{m,\ell}^{(n)}+1 ) } } \left( \sqrt{K_{m,\ell}^{(n)}} e^{j \psi_{m,\ell}^{(n)}} +   {a}_{m,\ell}^{(n)}  \right),
\end{align}
where $\beta_0$ and $\kappa$ are, respectively, the pathloss at a reference distance of $1$~m and the pathloss exponent; in turn, the distance is $d_{m,\ell}^{(n)}$ and the Rician factor is
\begin{equation}
K_{m,\ell}^{(n)} = A_1 \exp \! \left( A_2  \arcsin \! \left(\frac{h_\ell - h_m^{(n)}}{d_{m,\ell}^{(n)}}\right) \right)
\end{equation}
for environment-dependent parameters $A_1$ and $A_2$ \cite{ChMeas}. The phase of the LoS component, $\psi_{m,\ell}^{(n)}$, is uniformly random to reflect drifting \cite{8952782, demir2020joint} whereas  the  small-scale fading is
\begin{equation}
{a}_{m,\ell} ^{(n)}\sim \mathcal{N}_{\mathbb{C}}(0,1).
\end{equation}
Hence, the channel power gain is given by
\begin{align}
    {r}_{m,\ell}^{(n)} 
    = \frac{\beta_0 }{\big( d^{(n)}_{m,\ell} \big)^\kappa}.
\end{align}

\subsection{Channel Acquisition}
{Unlike most of the UAV-related literature, this paper incorporates imperfect channel estimates to enhance the model's realism.} Given that the system under consideration features a small number of UAVs,  orthogonal pilots can be allocated to each of those UAVs while keeping the overhead at bay;
pilot contamination is thus not an issue.
After observing the length-$\tau$ pilot sequence transmitted by the $m$th UAV, the $\ell$th AP can compute the MMSE channel estimate $\hat{g}_{m,\ell}^{(n)}$, which  satisfies
\begin{equation}
    {g}_{m,\ell}^{(n)} = \hat{g}_{m,\ell}^{(n)} + \tilde{g}_{m,\ell}^{(n)}
\end{equation}
such that $\hgml^{(n)}$ is zero-mean with \cite{lozano2008interplay,996869}
\begin{align}
    \gamma_{m,\ell}^{(n)} & = \mathbb{E}\big\{ | \hat{g}_{m,\ell}^{(n)} |^2  \big\} \\
    & = \frac{\big(r_{m,\ell}^{(n)}\big)^2}{  {r}_{m,\ell}^{(n)} +  \frac{\sigma^2}{p^{\text{t}}\tau }},
\end{align}
given $p^{\text{t}}$ as the pilot transmit power and $\sigma^2$ as the noise power at the receiver. Finally, the estimation error $\tilde{g}_{m,\ell}^{(n)}$ is zero-mean with variance 
\begin{equation}
c_{m,\ell}^{(n)} ={r}_{m,\ell}^{(n)}- \gamma_{m,\ell}^{(n)}.
\end{equation}

\subsection{Energy Consumption Model}\label{Sec_Fire:EnergyCons}

{
UAVs are governed by certain dynamics relating location, speed, and acceleration. A first-order Taylor expansion yields [45],[46]
    \begin{align}\label{eq:linloc}
        \boldsymbol{q}_m^{(n+1)} = \boldsymbol{q}_m^{(n)} + \boldsymbol{v}_m^{(n)}\delta + \frac{1}{2}\boldsymbol{a}_m^{(n)}\delta^2,
    \end{align}
    and
    \begin{align}\label{eq:linvel}
        \boldsymbol{v}_m^{(n+1)} = \boldsymbol{v}_m^{(n)} + \boldsymbol{a}_m^{(n)}\delta,
    \end{align}
    where $\boldsymbol{v}_m^{(n)}$ and $\boldsymbol{a}_m^{(n)}$ denote the speed and acceleration vectors.}
For a quad-rotor UAV moving in 3D with a given climb angle $\tau_m^{(n)}$, {and denoting the magnitude of the speed vector by $v_m^{(n)}$, i.e., $v_m^{(n)} =  \| \boldsymbol{v}_m^{(n)} \|$}, the power consumption can be modeled as \cite{Zeng2019EnergyUAV,9171468}
\begin{align}\label{eq:energy_cons}
    P_m^{(n)} & = P_0 \Bigg(1+ \frac{3v_m^{(n) 2}}{U_{\rm tip}^2}\Bigg) + P_i \Bigg( \sqrt{1 + \frac{v_m^{(n) 4}}{4v_0^4}} - \frac{v_m^{(n) 2}}{2v_0^2} \Bigg)^{\! 1/2} \nonumber \\
    & \quad + \frac{1}{2}d_0 \iota s A v_m^{(n) 3} + mgv^{(n)}_m \sin \tau_m^{(n)},
\end{align}
where $P_0$ and $P_i$ represent the blade profile and induced powers, respectively, $U_{\rm tip}$ is the tip speed of the rotor blade, and $v_0$ is the mean rotor induced velocity. In addition, $d_0$ and $s$ are   the fuselage drag ratio and rotor solidity,
respectively, while $\iota$ denotes the air density and $A$ the rotor disc area.

As advanced, the UAV's mission is divided into {tracking} and {charging} stages. The switching
occurs at $n=n_m$, once the UAV's energy falls below a threshold  that is a function of the distance between the UAV and the closest charging point (see Sec. \ref{Sec:SimRes_Fire}). Consequently, the UAV needs to head for a charging station and the remaining energy, denoted by $E_{m}^{(n_m)}$, must satisfy

\begin{align}\label{ct:energy}
    E_{m}^{(n_m)} - \!\!\!\!  \sum \limits_{n=n_{m}}^{n_{m}+N_{c_m}} \!\! P_m^{(n)} \delta  \geq 0,
\end{align}
to avoid running out of energy, where $N_{c_m}$ is the number of  slots it takes for the $m$th UAV to reach charging station $\boldsymbol{c}_c$.

Other constraints include the maximum and minimum UAV altitudes, a maximum velocity and acceleration, and the avoidance of collisions:
\begin{align}\label{ct:altitude}
    h_{\rm min}  \leq h_m^{(n)} & \leq h_{\rm max} \\
\label{ct:velBS}
    {v}_m^{(n+1)} & \leq V_{\rm max} \delta \\
\label{ct:accel}
    \| \boldsymbol{v}_m^{(n+1)} - \boldsymbol{v}_m^{(n)} \| & \leq A_{\rm max} \delta \\
\label{ct:collision_av}
    {\rm d}_{m,j}^{(n)} & \geq D_{\rm safe},
\end{align}
where $V_{\rm max}$ and $D_{\rm safe}$ are the maximum UAV velocity and  minimum  safety distance among UAVs, respectively, ${\rm d}_{m,j}^{(n)}$ is the distance between the $m$th and $j$th UAVs, { and $A_{\rm max}$ is the maximum acceleration}.

\section{Cell-free Connectivity}\label{Sec:CF_Fire}

On a given time-frequency resource unit, the number of active APs is $L^{(n)}$, which is a function of time. The uplink channel matrix is then
\begin{equation}
    \Gm^{(n)} = \begin{pmatrix}
        \boldsymbol{g}_{1}^{(n)} , \dots , \boldsymbol{g}_M^{(n)} 
  \end{pmatrix},
\end{equation}
where $\boldsymbol{g}_m^{(n)} \in \mathbb{C}^{L^{(n)} \times 1}$ is the channel between the $m$th UAV and the active APs, satisfying
\begin{equation}
    \Gm^{(n)} = \Gmh^{(n)} + \Gmt^{(n)},
\end{equation}
where $\Gmh^{(n)}$ and $\Gmt^{(n)}$ are the channel estimate and error matrices, respectively. {In general, subsets of APs can provide service to subsets of UAVs \cite{9043895,9930941,10186347}, yet this work considers full connectivity, for the sake of the simplicity of notation,} pooling the observations from the $L^{(n)}$ APs into
\begin{align}\label{eq:complete}
    \yb^{(n)} & =  \Gm^{(n)} \x^{(n)} +\boldsymbol{n}^{(n)}  \\
    & = \underbrace{ \Gmh^{(n)} \x^{(n)}}_\text{signal}  + \underbrace{ \Gmt^{(n)}  \x^{(n)} +  \boldsymbol{n}^{(n)}}_\text{ effective noise: $\boldsymbol{n}_{\rm e}^{(n)}$ },
\end{align}
where
\begin{equation}
\x^{(n)} = \left( \sqrt{p_1^{(n)}}s_1^{(n)},\dots, \sqrt{p_M^{(n)}} s_M^{(n)} \right)^{\! \mathrm{T}},
\end{equation}
with unit power symbols $s_m^{(n)}$ while $p_m^{(n)}$ denotes the transmit powers and the noise is  $\boldsymbol{n}^{(n)} \sim \mathcal{N}_{\mathbb{C}}(\boldsymbol{0},\sigma^2 \boldsymbol{{I}})$. The effective noise $\boldsymbol{n}_{\rm e}^{(n)}$ has covariance $\boldsymbol{\Sigma}^{(n)} 
= \boldsymbol{{D}}^{(n)} + \sigma^2 \boldsymbol{{I}}$ given
\begin{align}
    \boldsymbol{{D}}^{(n)} & = \mathbb{E} \! \left \{ \big( \Gmt^{(n)} \x^{(n)} \big)\big( \Gmt ^{(n)} \x^{(n)} \big) ^* \right \}  \\
    & = 
    \text{diag} \! \left\{ \sum \limits_{m=1}^M c_{m,1}^{(n)}p_m^{(n)} , \dots , \sum \limits_{m=1}^M c_{m,L}^{(n)}p_m^{(n)} \right \}.  
\end{align}
Upon observing $\boldsymbol{y}^{(n)}$,  the combiner that maximizes the signal-to-interference-plus-noise ratio (SINR) is the MMSE filter, achieving a value of \cite{9043895} 
\begin{align}
    \mathrm{SINR}_m^{(n)}
    & = \boldsymbol{\hat{g}}_m^{(n) *} \left(   \sum \limits_{n \neq m} \boldsymbol{\hat{g}}_{n}^{(n)}\boldsymbol{\hat{g}}_{n}^{(n) *} p_n^{(n)} + \boldsymbol{\Sigma}^{(n)} \! \right)^{\!\! -1} \!\!\! \boldsymbol{\hat{g}}_m^{(n)} \, p_m^{(n)},\label{eq:SINR}
\end{align}
giving a spectral efficiency of
\begin{equation}\label{eq:SEff}
    \mathrm{SE}_m^{(n)} = \left(1 - \frac{\tau}{\tau_{\rm c}} \right) \mathbb{E}\{ \log_2 ( 1 + \mathrm{SINR}_m^{(n)} )\},
\end{equation}
where $ \tau / \tau_{\rm c}$ accounts for the pilot overhead {and $ \tau_c$ denotes the total number of resource units within a coherence block}. Random matrix theory offers a way to circumvent the numerical evaluations of the above expectation, providing stable spectral efficiency forms that depend solely on large-scale parameters. 

\subsection{Large-Dimensional Analysis}\label{Sec:LargeDimFDMA}

When evaluating \eqref{eq:SEff} for $M,L^{(n)}\to \infty$, convergence to nonrandom limits
is assured provided that
\begin{align}
    \boldsymbol{\Gamma}_m^{(n)} & = \mathbb{E} \Big\{  \boldsymbol{\hat{g}}_{m}^{(n)} \boldsymbol{\hat{g}}_{m}  ^{(n) *} \Big\}  \\
     & = \mathrm{diag} \big \{ \gamma_{m,\ell}^{(n)} \; \forall \ell \big \}
\end{align}
satisfies some technical conditions. Specifically, the inverse of the resolvent matrix in \eqref{eq:SINR} must exist, which is ensured by the presence of $\boldsymbol{\Sigma}^{(n)}$, while $\boldsymbol{\Gamma}_m^{(n)}$  must have uniformly bounded spectral norm, meaning that the received power does not concentrate on a subset of dimensions as the network grows large.

\begin{theorem}\label{prop:MMSE_Fire}
For  $M,L^{(n)}\to \infty$ with MMSE reception, $\mathrm{SINR}_m^{(n)} - \overline{\mathrm{SINR}}_m^{(n)} \to 0$ almost surely (a.s.) with
\begin{equation}\label{eq:rmt_mmse_fire}
    \overline{\mathrm{SINR}}_m^{(n)} =  \sum \limits_{\ell=1}^{L^{(n)}}    \frac{\gamma_{m,\ell}^{(n)} p_m^{(n)}}{\sum \limits_{i \neq m}\frac{\gamma_{i,\ell}^{(n)}}{1 + e_{i}}p_i^{(n)} +  \sum \limits_{\forall i}c_{i,\ell}^{(n)}p_i^{(n)} +   \sigma^2}.
\end{equation}
The coefficients $e_{j}$ are obtained iteratively with $e_{j} = \lim_{s \to \infty} e_{j}^{(s)}$, $e_{j}^{(0)} = L^{(n)}$, and 
\begin{equation}
    e_{j}^{(s)} = p_j^{(n)} \, \mathrm{tr} \Bigg[ \boldsymbol{\Gamma}_j^{(n)} \bigg(  \sum \limits_{i \neq j }^M \frac{\boldsymbol{\Gamma}_i^{(n)}}{1+e_{i}^{(s-1)}} \, p_i^{(n)}  +  \boldsymbol{\Sigma}^{(n)} \bigg)^{\!-1}  \Bigg] .
\end{equation}
\end{theorem}
\begin{proof}
See Appendix \ref{proof:SINR_Fire}.
\end{proof}

As anticipated, \eqref{eq:rmt_mmse_fire} is a deterministic quantity that depends only on large-scale parameters.
From the continuous mapping theorem \cite{Mann1943OnSL}, it follows that
\begin{equation}
    \mathrm{SE}_m^{(n)}- \mathrm{\overline{SE}}_m^{(n)} \stackrel{a.s.}{\to} 0,
\end{equation}
where
\begin{equation}
\mathrm{\overline{SE}}_m^{(n)} = \left(1 - \frac{\tau}{\tau_{\rm c}} \right)   \log_2 \! \left( 1 +  \overline{\mathrm{SINR}}_m^{(n)} \right)    
\end{equation}
is the asymptotic spectral efficiency.

\section{Problem Formulation}\label{Sec:Problem_Fire}

Let us now turn to optimizing the UAV trajectories and the transmit powers.
In the tracking stage, the UAVs in the set $\mathcal{S}_1^{(n)}$ have sufficient energy to track while, in the charging stage, the UAVs in the set $\mathcal{S}_2^{(n)}$ attempt to reach one of the $C$ charging points as expeditiously as possible.
Given the different nature of the two stages, separate optimizations  are devised.

\subsection{Tracking Stage}

During the tracking stage, UAVs periodically capture images, to be forwarded to the APs every $N$ slots. To keep the information fresh on average, the received image bit content must satisfy
\begin{align}\label{ct:Rate}
    \delta W \mathrm{\overline{SE}}_m^{(n)} \geq \frac{B}{N}
\end{align}
where $B$ is given in \eqref{eq:n_bits_img} and $W$ is the transmission bandwidth. Note that this conservatively assumes that successive transmissions involve completely new images.

With (\ref{ct:Rate}) as a constraint that adds to those on altitude, speed, and collision-avoidance,
the optimization of trajectories and powers can be cast as the minimization of (a weighted version of) the area per pixel.
As discussed in Sec. \ref{sec:Camera}, the area per pixel reflects the tradeoff between image quality and FoV.
Before proceeding, though, the definition of area per pixel must be extended from a  single UAV to multiple UAVs, as their FoV may intersect. 
If multiple cameras cover a certain point $\boldsymbol{v}$,  the area per pixel is \cite{5959179}
\begin{align}\label{eq:mult_camera}
    \tilde{f}\Big( \qm^{(n)}: m \in \mathcal{S}_1^{(n)}, \boldsymbol{v} \Big) = \frac{1} 
{ \sum_{i \in \mathcal{N}_{\boldsymbol{v}}} f\big(\boldsymbol{q}_i^{(n)}, \boldsymbol{v}\big)^{-1} + \Delta },
\end{align}
where $\mathcal{N}_{\boldsymbol{v}}$ is the set of UAVs whose cameras cover $\boldsymbol{v}$ while $\Delta$ is a regularization factor that prevents the function from diverging if none of the UAVs covers $\boldsymbol{v}$. {Therefore, the joint optimization of the UAVs in tracking mode}  emerges as
\begin{equation}\label{opt:ProblemTrack}
    \begin{aligned}
    \min_{\qm^{(n)}, p_m^{(n)} } 
    & \int_{\mathcal{F}}  \tilde{f}\Big( \qm^{(n)}: m \in \mathcal{S}_1^{(n)}, \boldsymbol{v} \Big)  \psi(\boldsymbol{v},n) d\boldsymbol{v}, \\ 
    \textrm{s.t.} \quad  &  p_m^{(n)} \leq p_{\rm max} \\
    & \eqref{eq:linloc}, \eqref{eq:linvel}, \eqref{ct:altitude}\text{-}\eqref{ct:collision_av}, \eqref{ct:Rate}
    \end{aligned}
\end{equation}
where 
$\psi : \mathbb{R}^{2} \times \mathbb{N}\rightarrow [0,\infty)$ is a density that weights the points on the FoV, assigning them relative importance (see Sec. \ref{sec:FireDyn}). As reported in Sec. \ref{Sec:SimRes_Fire}, the solution to the above problem also maximizes the coverage, defined as {the ratio between the points with positive density that lie inside the UAVs' FoV and the total number of points with positive density: } 
\begin{align}
    {\rm Coverage}^{(n)} = \frac{ \int_{\mathcal{F}}  \mathbb{I} \! \left \{ \psi(\boldsymbol{v},n) > 0 \right \} \mathbb{I} \! \left \{ \exists \mathcal{B}_m^{(n)}: \boldsymbol{v} \in \mathcal{B}_m^{(n)}  \right \} d\boldsymbol{v}}{\int_{\mathcal{F}}  \mathbb{I} \! \left \{ \psi(\boldsymbol{v},n) > 0 \right \} d\boldsymbol{v} },\label{eq:cov}
\end{align}
where $\mathbb{I}\{\cdot\}$ is an indicator function whose value is one when the argument's condition is true and zero otherwise.
\subsection{Charging Stage}

During the charging stage, a UAV tries to reach one of the $C$ charging stations. For the $m$th UAV, the closest one is
\begin{align}\label{opt:ch_point}
    c_m = \arg \min_c \|\boldsymbol{q}_m^{(n_{m})} - \boldsymbol{c}_c \| ,
\end{align}
where $n_{m}$, recall, is the time slot in which the $m$th UAV switches from tracking to charging. {In an unconstrained setup, trivial straight-line trajectories would be optimal, but the actual solution is subject to constraints.} In particular, choosing to minimize the energy consumption, 
the UAV trajectories are optimized via 
\begin{equation}\label{opt:ProblemCharge}
    \begin{aligned}
    \min_{\qm^{(n)}, N_{c_m}}  
    & \ls   \sum \limits_{n=n_{m}}^{n_{m}+N_{c_m}} P_m^{(n)} \delta,    \\ 
    \textrm{s.t.} \quad \eqref{eq:linloc}, \eqref{eq:linvel}, &   \eqref{ct:energy} \text{-} \eqref{ct:collision_av} \\
    \end{aligned}
\end{equation}
where $N_{c_m}$ is introduced in Sec. \ref{Sec_Fire:EnergyCons}, $ P_m^{(n)}$ is defined in \eqref{eq:energy_cons}, and the constraints ensure sufficient energy to reach a charging point, the avoidance of collisions, and the respect of altitude and velocity limits. The switching between tracking and charging, recall, is controlled by the remaining energy  at the UAV, as furthered 
in Sec. \ref{Sec:SimRes_Fire}. 

\subsection{Fire Dynamics}\label{sec:FireDyn}

The density function $\psi(\boldsymbol{v},n)$ depends on the event being tracked {and can be generated either synthetically through a model or constructed from actual data}. While the solution presented in the next section is valid for any generic density function, for wildfires this function should encode certain features that model the spread of a fire. Wildfire simulation is an active field of research in itself \cite{FireCit}. FARSITE, an established model employed by government agencies, is adopted here
\cite{Finney_1998}. It dictates that the spread of every ignition  originally follows an ellipse. Subsequently, the points on the ellipse serve as new fronts as per Huygens principle \cite{Finney_1998}, each such front point growing as a new ellipse. The set of new front points is defined by the combination of the local ellipses, i.e., their convex hull \cite[Fig. 1]{Finney_1998}. Because of the different local weather conditions, fuels, or terrain, the dimensions of the local ellipses can be different and result in a new front that is no longer elliptic.
Let us define $U^{(n)}$ [m/s] and $\theta_{\text{wind}}^{(n)}$ as the midflame wind's speed and direction, respectively, modelled as
\begin{align}
U^{(n)} & \sim |\mathcal{N}(U_0, \sigma_U) | \\
\theta_{\text{wind}}^{(n)} &  \sim \mathcal{N}(\overline{\theta}_{\text{wind}}, \sigma_\text{wind}).
\end{align}
Generating the ellipse corresponding to each front point requires its minor and major axes, $2a^{(n)}$ and $2b^{(n)}$, respectively, with
\begin{align}
    a^{(n)} & = \frac{1}{2 \, \mathsf{LB}^{(n)}} \left( \mathsf{R} + \frac{\mathsf{R}}{\mathsf{HB}^{(n)}} \right) \\
    b^{(n)} & = \frac{1}{2 } \left( \mathsf{R} + \frac{\mathsf{R}}{\mathsf{HB}^{(n)}} \right),
\end{align}
where $\mathsf{R}$ [m/min] is the fire's steady-state spreading rate and 
\begin{align}
    \mathsf{LB}^{(n)} & = 0.936 \, e^{0.2566U^{(n)}} + 0.461 \, e^{–0.1548U^{(n)}} – 0.397 \\
    \mathsf{HB}^{(n)} & = \frac{\mathsf{LB}^{(n)} + \sqrt{ [ \mathsf{LB}^{(n)}]^2 – 1 } }{\mathsf{LB}^{(n)} \, – \, \sqrt{ [\mathsf{LB}^{(n)}]^2 – 1}} .
\end{align}
Hence, at time $n+1$, the $i$th front point, denoted by $\boldsymbol{z}_i^{(n)} \in \mathbb{R}^2$, generates the ellipse  
\begin{align}\label{eq:ellipsep}
     \boldsymbol{z}_i^{(n)} +  \delta 
    \begin{bmatrix}
       c_x^{(n)}  \sin \theta_{\text{wind}}^{(n)} +  a^{(n)} \cos \omega\\
       c_y^{(n)}  \cos \theta_{\text{wind}}^{(n)} +  b^{(n)} \sin \omega
   \end{bmatrix},
\end{align}
where $0 \leq \omega \leq 2\pi$ whereas  $c_x^{(n)}$ and $c_y^{(n)}$ represent the fire spreading gradients, respectively. As the exact FARSITE implementation requires detailed information about the environment, and with a view to a general formulation,
we consider a simplified set of FARSITE parameters that has been widely used in the literature \cite{9147613,7587184}. In this simplified model \cite{Finney_1998}
\begin{align}
     c_x^{(n)}=c_y^{(n)}= \frac{\mathsf{R}}{2}  \left( 1 – \frac{1}{\mathsf{HB}^{(n)}} \right). 
\end{align}
With that, a 2D histogram can be produced to serve as the density function $\psi(\boldsymbol{v},n)$, assigning a nonzero weight to the points of the fire's perimeter being tracked. {Fig. \ref{fig:Farsite}  illustrates the expansion of a wildfire from an ignition point (black dot) using FARSITE. Assuming wind in the direction of the red arrow, an ellipse with minor and major axes, $2a^{(n)}$ and $2b^{(n)}$, respectively, represents the next fire front as shown in Fig. \ref{fig:Farsite}b. Then, each blue point in the new fire front acts like a new  ignition point and the convex hull of the corresponding ellipses generates the new  fire front in Fig. \ref{fig:Farsite}d. Here, the density function $\psi(\boldsymbol{v},n)$ corresponds to the perimeter, represented by the red curve. }

\begin{figure*}[!htb]
     \centering
     \includegraphics[scale=0.6]{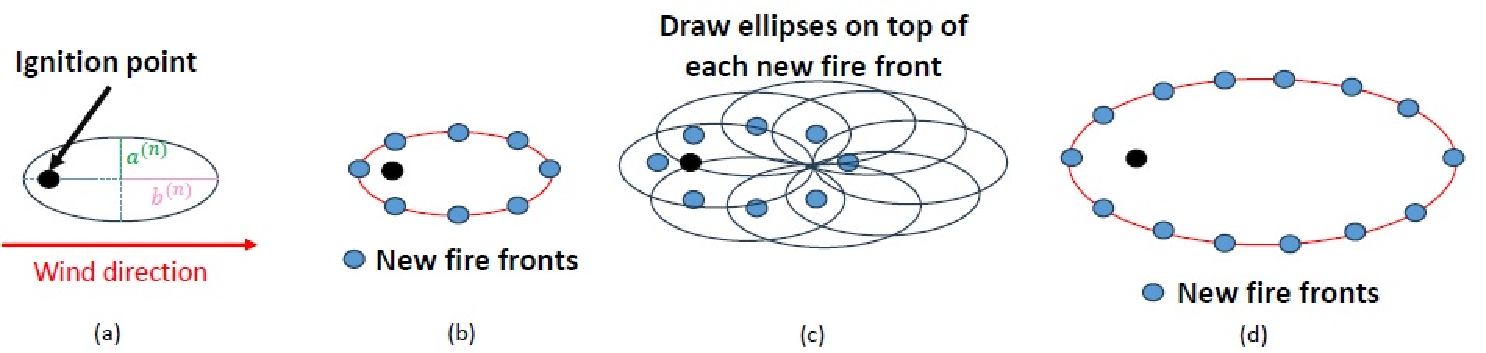}
     \caption{Summary of FARSITE fire propagation model.}
     \label{fig:Farsite}
\end{figure*}

\section{Proposed Solution}\label{Sec:Sol_Fire}
\subsection{DQL Fundamentals}
Solving \eqref{opt:ProblemTrack} and \eqref{opt:ProblemCharge} is challenging due to (\emph{i}) the variety of constraints, (\emph{ii}) the lack of convexity with respect to the optimization variables, and (\emph{iii}) the time-varying nature of the problem. However, the optimization can be formulated as a Markov decision process (MDP) problem, for which RL is an appropriate solution. In general, an MDP is defined by the tuple $(\mathcal{S},\mathcal{A},\mathcal{P},R,\gamma)$, where $\mathcal{S}$ and $\mathcal{A}$ are the state and action spaces, respectively. In turn, $\mathcal{P}$ denotes the state transition probability with $P \! \left(s^{(n+1)}|s^{(n)},a^{(n)}\right)$ being the probability of transitioning from state $s^{(n)}$ to $s^{(n+1)}$ after action $a^{(n)}$. Finally, $R$ is the reward function and $\gamma \in [0,1]$ denotes the discount factor in such function. According to the MDP,
\begin{align}
    P \! \left(s^{(n+1)}|s^{(0)},a^{(0)},\ldots,  s^{(n)},a^{(n)}\right) = P \!
 \left(s^{(n+1)}|s^{(n)},a^{(n)}\right) \nonumber
\end{align}
and
\begin{align}
    R \! \left(s^{(n+1)}|s^{(0)},a^{(0)},\dots,  s^{(n)},a^{(n)} \right) & = R 
\! \left(s^{(n+1)}|s^{(n)},a^{(n)}\right) \nonumber \\
    & = r \! \left(s^{(n)},a^{(n)}\right),
\end{align}
where $r(\cdot)$ is introduced as shorthand notation. The goal in RL is to find a policy $\pi: \mathcal{S} \rightarrow \mathcal{A}$ that maximizes the average discounted reward,
\begin{align}\label{eq:opt_policy1}
    \pi^* = \arg \max_{\pi} \mathbb{E} \! \left \{ \sum \limits_{j=0}^\infty \gamma^{j} r \! \left(s^{(j)},a^{(j)}\right) \bigg| \pi  \right \},
\end{align}
with expectation over $s^{(0)} \sim P \! \left(s^{(0)}\right)$, $a^{(n)} \sim \pi \! \left(\cdot | s^{(n)}\right)$ and $s^{(n+1)} \sim P \! \left(s^{(n+1)}|s^{(n)},a^{(n)}\right)$.
In RL, the Q-function measures the expected future reward when the system is in state $s$, performs action $a$, and follows policy $\pi$, namely
 \begin{align}
     Q^{\pi}(s,a) = \mathbb{E}_{\pi} \! \left \{ \sum \limits_{j=0}^\infty \gamma^{j} r \left(s^{(j')},a^{(j')} \right) \bigg| s^{(n)} = s, a^{(n)} = a  \right \}.
 \end{align}
where $j'= j+n+1$. An interesting property of the Q-function is that every optimum policy  $\pi^*$ achieves the optimum Q-function, i.e.,
\begin{align}
    Q^{\pi^*\!}(s,a) = Q^*(s,a),
\end{align}
where $*$ denotes optimality. Hence, one can identify the optimum policy through the Q-function. {In fact, through dynamic programming and the Bellman equation, learning the Q-function is possible. Under the temporal difference (TD)
learning solution, which combines the current estimate of the Q-function with samples obtained from the environment, the estimates can be iteratively updated following \cite{sutton2018reinforcement}}
\begin{align}\label{eq:q-fun_up}
    Q(s,a) \leftarrow (1-\alpha) Q(s,a) + \alpha \left[ r(s,a) + \gamma \max_{a'}Q(s',a') \right].
\end{align}
Although \eqref{eq:q-fun_up} provides a straightforward rule to learn the Q-function, the corresponding learning process is highly dependent
on the state and action space dimensions, as every pair $(s,a)$ needs to be explored many times.
Multi-UAV communication problems, with their complex and dynamic environments, and high-dimensional state and action spaces, motivate the use of DQL. 
This leverages deep neural networks to approximate the Q-value function, with a slight abuse of the notation, as $Q(s,a) \approx Q(s,a ;\boldsymbol{\theta})$ where 
 $\boldsymbol{\theta}$ are the neural network parameters tuned to minimize
\begin{align}\label{eq:loss_DNN}
    L(\boldsymbol{\theta} ) = \mathbb{E} \! \left\{ | y(s,a) - Q(s,a;\boldsymbol{\theta}) |^2 \right \},
\end{align}
where $y(s,a)$ is the target value defined as
\begin{align}\label{eq:out_DNN}
    y(s,a) = r(s,a) + \gamma \max_{a'}Q(s',a'; \boldsymbol{\theta}).
\end{align}
Still, DQL suffers from a major limitation as it considers a discrete action space.
To 
consider continuous  actions, we adopt TD3, a policy-based 
model-free algorithm \cite{fujimoto2018addressing}. TD3 leverages an actor-critic architecture, with both modules composed by neural networks. On the one hand, the critic aims at learning the Q-function, which is exploited by the actor to know how beneficial an action is. On the other hand, the actor's role is to learn a policy that, as opposed to traditional Q-learning, outputs continuous action values. Such policy is parameterized by $\boldsymbol{\phi}$,
the actor's neural network parameters, that are updated by the gradient computed as
\begin{align}\label{eq:PG_Grad}
    \nabla J(\boldsymbol{\phi}) = \mathbb{E} \! \left \{\nabla_{\boldsymbol{\phi}} \pi_{\boldsymbol{\phi}}(s)\  \nabla_a Q(s,a;\boldsymbol{\theta}) |_{a=\pi_{\boldsymbol{\phi}}(s)}  \right \},
\end{align}
where
\begin{align}
    J(\boldsymbol{\phi}) = \mathbb{E}_\pi \! \left \{ \sum \limits_{j=0}^\infty \gamma^{j} r \! \left(s^{(j)},a^{(j)}\right)  \right \} .
\end{align}
In addition, DQL algorithms typically include the so-called target networks. These are copies of the original networks whose parameters remain frozen over a number of iterations. Then, a soft-update rule is applied. Concretely, we denote by $\boldsymbol{\theta}'$ and $\boldsymbol{\phi}'$ the critic and actor neural network target parameters, respectively, with update rules
\begin{align}\label{eq:update1}
    \boldsymbol{\theta}^{'} \leftarrow (1-\tau_{\rm T}) \boldsymbol{\theta}^{'} + \tau_{\rm T}\boldsymbol{\theta},
\end{align} 
\begin{align}\label{eq:update2}
    \boldsymbol{\phi}^{'} \leftarrow (1-\tau_{\rm T}) \boldsymbol{\phi}^{'} + \tau_{\rm T}\boldsymbol{\phi},
\end{align}
where $\tau_{\rm T}$ controls the  memory of the updates. The inclusion of target networks addresses the moving target issue that arises when Q-values and target values are obtained using the same neural network and are constantly changing, creating a moving target for the learning; the algorithm's goal then keeps varying as the agent learns, hampering the finding of a solution.

Relative to previous actor-critic methods,  TD3  introduces improvements that enhance the stability during learning:

\begin{enumerate}
\item It features two critics, and respective target critic networks. This addresses the overestimation of Q-values caused by the use in previous methods of the maximum action value to approximate the maximum expected action value, as per \eqref{eq:out_DNN}. In TD3, the agent selects the Q-target network with smaller Q-value to construct the target value as
    \begin{align}
        y(s,a) = r(s,a) + \gamma \min_{i=1,2}Q(s',a';\boldsymbol{\theta}_i').
    \end{align}

\item To reduce the accumulation of residual errors, the policy network parameters are updated less frequently than the Q-function parameters. Precisely, \cite{fujimoto2018addressing} suggests a policy update for every two Q-function updates. 

\item To prevent overfitting to specific estimated values, TD3 adds clipped random noise to the target actions,
\begin{align}\label{eq:a_noisy}
        \hat{a} = \pi_{\boldsymbol{\phi}'}(s) + \hat{\epsilon},
    \end{align}
    where $\hat{\epsilon} \sim \text{clip} \big(\mathcal{N}(0,\hat{\sigma}_a),-\hat{\epsilon}_{\rm max},\hat{\epsilon}_{\rm max} \big)$, i.e., the values outside the interval $[-\hat{\epsilon}_{\rm max},\hat{\epsilon}_{\rm max}]$ are clipped to the interval edges. 
\end{enumerate}

TD3 is applicable to  \eqref{opt:ProblemTrack} and \eqref{opt:ProblemCharge}. Next, we define the states, actions, and rewards for each problem with one agent coordinating a swarm of $M$ UAVs. 

\subsection{TD3: Tracking Stage}

To solve \eqref{opt:ProblemTrack}, {single-agent or multi-agent algorithms can be employed. A single-agent solution would coordinate the entire swarm and present a significant challenge due to the exponential growth of state and action spaces with the number of UAVs. Alternatively, large MDPs can be factored into simpler ones, leading to simpler agents and distributed solutions \cite{NIPS2001_7af6266c,factoredmdp}. Accordingly, we use the multi-agent solution, where each UAV is controlled by a different agent. } The following is needed:
\begin{itemize}
    \item The state space, which includes information about (\emph{i}) the UAV locations, (\emph{ii}) the portion of the image bit content that has not yet been transmitted, and (\emph{iii}) the state of the fire coverage. 
    To quantify the latter, let us compute the center of mass of the area per pixel as per \eqref{eq:masscent}.
     \begin{figure*}
        \begin{align}\label{eq:masscent}
            \big( {\rm x}_{\rm c}^{(n)}, {\rm y}_{\rm c}^{(n)} \big ) = \frac{\int_{\mathcal{F}}  \boldsymbol{v}\tilde{f}\Big( \qm^{(n)}: m \in \mathcal{S}_1^{(n)}, \boldsymbol{v} \Big)  \psi(\boldsymbol{v},n) d\boldsymbol{v}}{\int_{\mathcal{F}}  \tilde{f}\Big( \qm^{(n)}: m \in \mathcal{S}_1^{(n)}, \boldsymbol{v} \Big)  \psi(\boldsymbol{v},n) d\boldsymbol{v}}
        \end{align}
    \end{figure*}
If important points remain uncovered,  the agent learns to move in that direction. The untransmitted portion of the image bit content is included in the state space as
    \begin{align}
        i_m^{(n)} = \max \bigg\{0, \bigg(1 - \frac{\delta W \mathrm{\overline{SE}}_m^{(n)}}{B/N} \bigg) \bigg\} .
    \end{align}
Altogether, the state of the $m$th UAV at time $n$ is
\end{itemize}
    \begin{align}\label{eq:stateTot}
        s_m^{(n)} = \left \{ \frac{x_{m}^{(n)}}{\rm S}, \frac{y_m^{(n)}}{\rm S}, \frac{h_m^{(n)}}{h_{\rm max}} , \frac{\boldsymbol{v}_m^{(n)}}{v_{\rm max}}, \frac{{\rm x}_{\rm c}^{(n)}}{\rm S}, \frac{{\rm y}_{\rm c}^{(n)}}{\rm S} , \frac{{\rm d}_{m,j}^{(n)}}{C_{\rm d}},i_{m}^{(n)} \right\},
    \end{align}
\begin{itemize}
    \item[] where ${\rm S}$ and $ C_{\rm d}$ are positive normalization constants, $\boldsymbol{v}_m^{(n)}$ is the UAV's speed and $ {\rm d}_{m,j}^{(n)}$ is the distance between the $m$th and the $j$th UAVs. 
   
    \item The action space defines the set of actions that the UAVs can take and  consists of variations in the 3D location and transmit power, such that $a_m^{(n)} \in \mathbb{R}^{4 \times 1}$. {Given the relationships established by \eqref{eq:linloc} and \eqref{eq:linvel}, the problem is solved with respect to the acceleration. The first three actions relate to the acceleration variations in 3D space, each constrained by a maximum and minimum value $[-A_{\rm max},A_{\rm max} ]$}. The last action component sets the transmit power, adjusted to any value satisfying $0 \leq p_m^{(n)} \leq p_{\rm max}$.
    
    \item The reward function
    evaluates the quality of the actions and motivates the agent to take those leading to desirable outcomes. Not only does it capture the cost function, but it satisfies the constraints. As per the reward shaping technique, the reward observed by the $m$th UAV is \cite{rew_shap}
    \begin{align}\label{eq:reward}
        r_m \! \left(s_m^{(n)},a_m^{(n)}\right) = \sum \limits_{w=1}^5 r_{m,w} \! \left(s_m^{(n)},a_m^{(n)} \right)  ,
    \end{align}
    whose terms are detailed next. In particular, $r_{m,1} \! \big( s_m^{(n)},a_m^{(n)} \big)$ contains the variation in the cost function. After performing action $a_m^{(n)}$, the state transitions to $s_m^{(n+1)}$, whose cost is defined as
    \begin{align}
        c(s_m^{(n+1)}) & = \int_{\mathcal{F}}  \tilde{f}\Big( \qm^{(n+1)}: m \in \mathcal{S}_1^{(n+1)}, \boldsymbol{v} \Big)
     \nonumber \\ 
        & \quad \cdot  
        \psi(\boldsymbol{v},n+1) d\boldsymbol{v}.
    \end{align}
    Thus, to motivate the UAVs to reduce the cost function, $r_{m,1} \! \left(s_m^{(n)},a_m^{(n)}\right)$ is defined as
    \begin{align}\label{eq:r1_tr}
        r_{m,1} \! \left(s_m^{(n)},a_m^{(n)}\right) = K_{\rm c}\left(1 - \frac{c(s_m^{(n+1)})}{I_{\rm r}^{(n)}}  \right) 
    \end{align}
   {where $K_{\rm c}>0$ and $I_{\rm r}^{(n)} = \frac{1}{\Delta}  \int_{\mathcal{F}} \psi(\boldsymbol{v},n+1) d\boldsymbol{v}$.}     Moreover, to avoid collisions among UAVs, we define 
    \begin{align}\label{eq:rewColl}
        r_{m,2}\left(s_m^{(n)},a_m^{(n)}\right) = \left\{\begin{array}{ll} 
                -K_{\rm coll} & \text{if $\exists$ } {\rm d}_{m,j}^{(n)}  \leq D_{\rm safe}\\
                0 & \text{otherwise} \\
                \end{array}, \right.
    \end{align}
    where $K_{\rm coll} > 0$. Similarly, we define the reward associated with \eqref{ct:Rate} as
    \begin{align}\label{eq:r3}
        r_{m,3} \! \left(s_m^{(n)},a_m^{(n)}\right) = \left\{\begin{array}{ll} 
                -K_{\rm f} & \text{if } \frac{B}{N} \geq \delta W \mathrm{\overline{SE}}_m^{(n)} \\
                0 & \text{otherwise} \\
                \end{array}, \right.
    \end{align}
    where $K_{\rm f}>0$.
    A similar expression is used to define $r_{m,4}(\cdot)$ and $r_{m,5}(\cdot)$, which penalizes actions taking the UAV out of limits and exceeding a maximum velocity, respectively, with respective penalties of $-K_{\rm h}$ and $-K_{\rm v}$. 
\end{itemize}

The {multi-agent} TD3 algorithm used to train a single UAV during the {tracking} stage  is summarized in Alg. \ref{alg:DRL_Monitor} whereas Fig. \ref{fig:DNN_1} presents its block diagram. {Note that the factorization of the larger problem into $M$ single-agent problems  allows for parallelization, leading to faster convergence and more efficient learning. Given that the UAVs in tracking mode have the same objective, their training will converge to the same models. Therefore, training can be conducted on one UAV, while the others can either follow the learned policy up to that point in time, perform random movements, or remain static.} Besides the initial network parameters, Alg.~\ref{alg:DRL_Monitor} accepts the number of episodes $\mathrm{E}$, which refers to the number of wildfire realizations the system will observe, as input.
Additionally, a replay buffer of size $|\mathcal{M}|$ stores transitions of the type $\left\{ s_m^{(n)}, a_m^{(n)}, r_m \! \left(s_m^{(n)},a_m^{(n)}\right), s_m^{(n+1)} \right\}$. Finally, $F$ refers to the update frequency over the policy and target networks.

\subsection{TD3: Charging Stage}

Since the {tracking} and {charging} stages are optimized separately, two independent models can be trained. The same ideas and algorithms described for tracking can be applied for charging, with only some modifications needed
to address \eqref{opt:ProblemCharge}. Concretely, during the {charging} stage,
only one UAV needs to be coordinated. Hence, the states consist of
\begin{align}\label{eq:statech}
    s_m^{(n)} = \left \{ \frac{{\rm x}_{u_m} - x_{m}^{(n)}}{\rm S}, \frac{{\rm y}_{u_m} -y_m^{(n)}}{\rm S}, \frac{{\rm h}_{u_m} -h_m^{(n)}}{h_{\rm max}} ,  {\rm d}_{m,j}^{(n)}, \frac{E_m^{(n)}}{E_0} \right\},
\end{align}
where, recall, the distance between the $m$th and $j$th UAVs is $ {\rm d}_{m,j}^{(n)}$ while $E_0$ is a positive constant normalizing the remaining energy. 
While $r_{m,2} (\cdot)$, $r_{m,4} (\cdot)$ and $r_{m,5} (\cdot)$ remain the same, $r_{m,3} 
\! \left(s_m^{(n)},a_m^{(n)}\right) = 0$ as the UAVs have 
no data to forward. The  main difference, though, is in the definition of $r_{m,1} \! \left(s^{(n)},a^{(n)}\right)$, in \eqref{eq:RewardCh} atop the next page, which includes the variation in the cost function:
\begin{itemize}
    \item When the charging point is reached, a positive reward $K_{\rm fin}$ is given.
    \item The term $- K_{\rm e} P_m^{(n)}\delta$ aims at minimizing the energy whereas the second component motivates the UAV to minimize the distance with the charging point. If it is approached, the UAV receives a positive reward; otherwise, the reward is negative. Finally, a negative reward of $K_{\rm en}$ is incurred if the remaining energy is exhausted.
\end{itemize}

\begin{figure*}
    \begin{align}\label{eq:RewardCh}
    r_{m,1}(s_m^{(n)},a_m^{(n)}) = \begin{cases}
       K_{\rm fin}  & \text{if $\boldsymbol{q}_m^{(n)} = \boldsymbol{c}_{c_m}$} \\  
      - K_{\rm e} P_m^{(n)}\delta +   K_{\rm d} \big( \| \boldsymbol{c}_{c_m}-\boldsymbol{q}_m^{(n)}\| - \| \boldsymbol{c}_{c_m}-\boldsymbol{q}_m^{(n+1)}\| \big)  & \text{otherwise}
 \end{cases}.
\end{align}
\end{figure*}

Alg. \ref{alg:DRL_Monitor} applies to the charging stage
with the above modifications in the rewards. In addition, whereas the {tracking} stage does not have a final state, the training of the {charging} model does reach a final state whenever $\boldsymbol{q}_m^{(n)} = \boldsymbol{c}_{c_m}$. Thus, the \textit{while} condition in the loop can be changed from $n < N_e$ to $\boldsymbol{q}_m^{(n)} \neq \boldsymbol{c}_{c_m}$.

\begin{figure*}[!htb]
     \centering
     \includegraphics[scale=0.7]{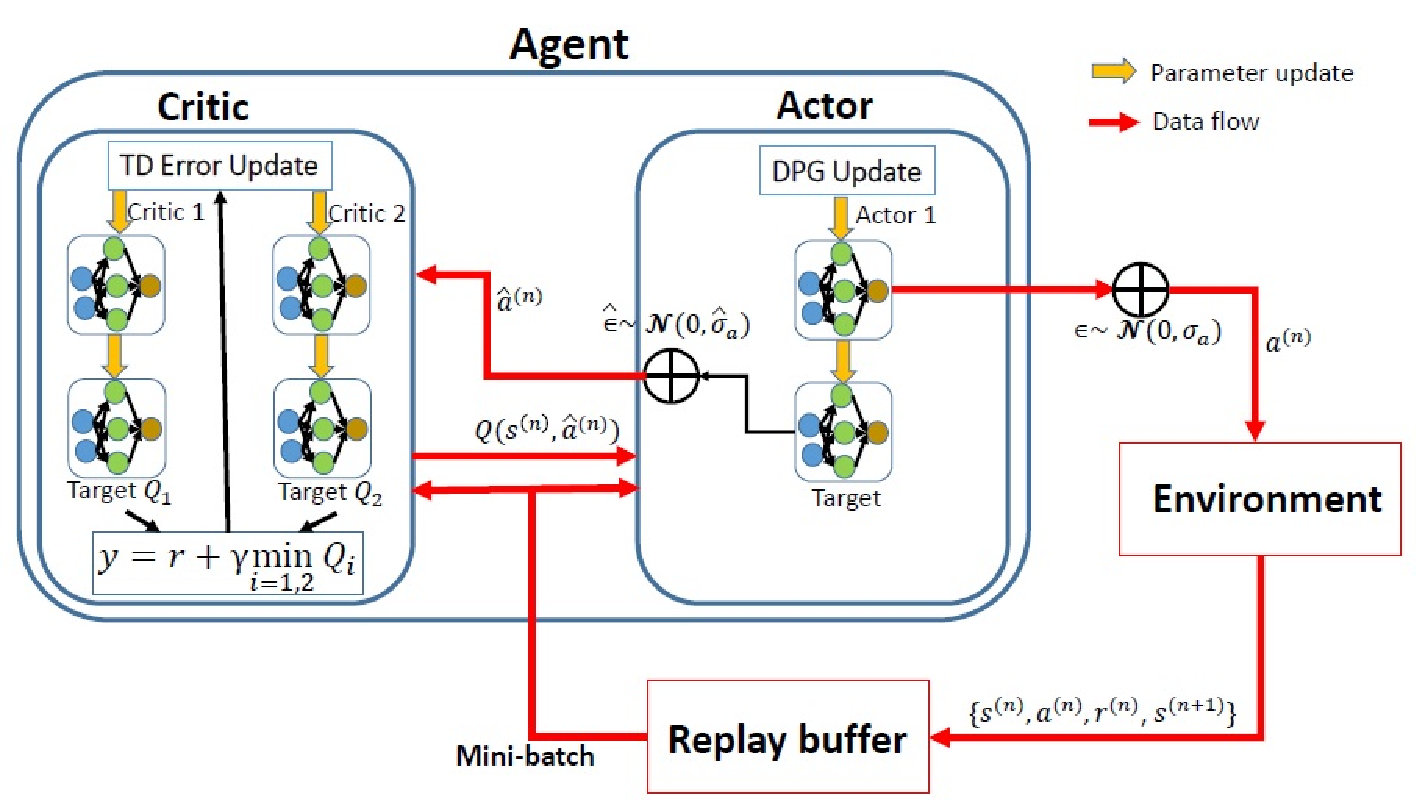}
     \caption{TD3 block diagram.}
     \label{fig:DNN_1}
\end{figure*}

\begin{algorithm}
\caption{TD3 - Tracking stage}\label{alg:DRL_Monitor}
\KwIn{No. of episodes  $\mathrm{E}$, initial policy parameters ${\boldsymbol{\phi}}$,
critic  parameters $\boldsymbol{\theta}_1$,  $\boldsymbol{\theta}_2$,  empty memory replay buffer $\mathcal{M}$ and parameter update frequency $F$. }
Initialize $\boldsymbol{\theta}_1^{'} \leftarrow \boldsymbol{\theta}_1$, $\boldsymbol{\theta}_2^{'} \leftarrow \boldsymbol{\theta}_2$, and $\boldsymbol{\phi}^{'} \leftarrow \boldsymbol{\phi}$,

\For {$e = 1,\dots,\mathrm{E}$}{
    Set $n=0$ and initialize UAV, AP and ignition point locations. \\
    \While{ $n < N_{e}$}{
        Obtain the current state $s^{(n)}$. \\
        Perform action $a_m^{(n)} = \text{clip}\left( \pi_{\boldsymbol{\phi}}\big(s_m^{(n)}\big) + \epsilon \right)$, with $\epsilon \sim \mathcal{N}(0,\sigma_a)$, and observe $s_m^{(n+1)}$.\\
        Compute the reward $r \! \left(s_m^{(n)},a_m^{(n)}\right)$ as in \eqref{eq:reward}. \\
         Store $\left\{ s_m^{(n)}, a_m^{(n)}, r_m \! \left(s_m^{(n)},a_m^{(n)}\right), s_m^{(n+1)} \right\} $ in  $\mathcal{M}$.\\
         Select a minibatch of $N_{\rm mem}$ experiences  $\left\{ s_m^{(n)}, a_m^{(n)}, r_m \! \left(s_m^{(n)},a_m^{(n)}\right), s_m^{(n+1)} \right\} $ from $\mathcal{M}$.\\
         Compute target actions $\hat{a}_m^{(n)}$ as in \eqref{eq:a_noisy}. \\
         Compute targets $y(s,a)$ as in \eqref{eq:out_DNN}. \\
         Update the critics by minimizing $\mathbb{E}\{ | y(s,a) - Q(s,a;\boldsymbol{\theta}_i) |^2  \}$ for $i=1,2$. \\
        \If{$n$ mod $F$}{
            Update $\boldsymbol{\phi}$ via \eqref{eq:PG_Grad}.\\
            Update  target networks $\boldsymbol{\theta}_1^{'}$, $\boldsymbol{\theta}_2^{'}$, and $\boldsymbol{\phi}^{'}$ as in \eqref{eq:update1} and \eqref{eq:update2}.    }
    $n = n + 1$. \\
    {Optional: move the  rest of UAVs. \\}
    }
}
\end{algorithm}

\section{Numerical Results}\label{Sec:SimRes_Fire}

For the purpose of performance evaluation, a $300$-m$^2$ environment is considered with the simulation parameters in Table~\ref{tab:SimParams_Fire}. {The camera parameters are set based on \cite{5959179,9802837}, emulating real devices,} while the UAV and cell-free parameters are borrowed from \cite{8097026,10186347,LOSLTE}. {Aiming at simulating real wildfires, the fire propagation components are chosen from  \cite{Finney_1998,9147613,7587184}}  whereas the TD3 parameters are configured as suggested in \cite{9426899,9508149,9504602,10086561}. Particularly, the variance of the noise added to the actions is ${\sigma}_a = \hat{\sigma}_a = 0.1$, $\hat{\epsilon}_{\rm max} = 0.5$, and the parameter soft updates are handled with $\tau_{\rm T} = 0.01$. Also, the learning rate for actor and critics is set to $5\cdot10^{-4}$ and $5\cdot10^{-3}$, respectively. The normalization constants in the states and rewards are obtained via cross-validation and summarized in Table \ref{Tab:NormCt}; other values could work as well. The initial UAV {and AP} locations, charging point coordinates, and fire ignition point are randomly generated. {At each time slot, the fire perimeter is updated according to Sec. \ref{sec:FireDyn} whereas the UAV locations are updated following the output of the actor block in the trained TD3 algorithm.} Finally, given that  we solve two separate problems, we first provide training and evaluation results assuming unlimited energy at UAVs. Then, after discussing the charging performance, the combined results of tracking and charging are presented.

\begin{table*}[h]
  \caption{Simulation parameters}\label{tab:SimParams_Fire}
  \centering
  \begin{tabular}{|l|l|l|l|l|l|}
    \hline
    Description & Parameter & Value & Description & Parameter & Value \\
    \hline
    Camera half-view angles & $\alpha_1$, $\alpha_2$ & 17.5$^{\circ}$, 13.125$^{\circ}$ & Maximum  power & $p_{\rm max}$   & 100 mW  \\
    Camera parameters & $a$, $b$ & $10^{-6}$, 10 &  Pathloss at $1$ m & $\beta_0$ & -30 dB \\
    Compression factor & $\rho$ & 0.4 & Noise power & $\sigma^2$ & -96 dBm \\
    No. of APs & $L$ & $10$  & Dense urban param. & $A_1$, $A_2$ &  0, 6.4 dB \\
    No. of charging points & $C$ & $4$  & Regularizer in \eqref{eq:mult_camera} & $\Delta$ &  $ 10^{-5}$ \\
    Carrier frequency & $f_{\rm c}$ & 2.4 GHz & Time slot &  $\delta$ & 0.5 s \\
     Maximum acceleration &  $A_{\rm max}$ & 1 & Transmission bandwidth &  $B$ & 10 MHz \\
     Maximum speed &  $V_{\rm max}$ & 20 & Learning rate &  $\gamma$ & 0.85 \\
     Minimum and maximum altitude &  $H_{\rm min}$, $H_{\rm max}$ & 100 m, 150 m & Image forwarding period &  $N$ & 2 \\
     Pilot sequence length &  $\tau$ & 200 &  No. of channel uses &  $\tau_{\rm c}$ & 6250  \\ 
     Mean midflame wind speed &  $U_{0}$ & 5 m/s & Minimum safety dist. &  $D_{\rm safe}$ & 4 m \\
     Midflame wind speed variance &  $\sigma_{U}$ & 1 m/s & Mean wind direction &  $\hat{\theta}_{\rm wind}$ & $\mathcal{U}[0,2\pi]$ \\
     Fire spreading rate &  $\mathsf{R}$ & 35 m/min & Wind direction variance &  $\sigma_{\rm wind}$ & 0.1  \\
    \hline
  \end{tabular}
\end{table*}

\begin{table}
\centering
\caption{State and reward parameters}
\begin{tabular} 
{ |l||l|l| }
 \hline
{\bf Description} & {\bf Parameter} & {\bf Value} \\
 \hline
    Normalization in \eqref{eq:stateTot}, \eqref{eq:statech} & ${\rm S}$ & 300  \\
    Constant in \eqref{eq:r1_tr} & $K_{\rm c}$ & 50 \\
    Penalty in \eqref{eq:rewColl} & $K_{\rm coll}$ & 100 \\
    Penalty in \eqref{eq:r3} & $K_{\rm f}$ &  15 \\
    Flying-out-of-limits penalty  & $K_{\rm h}$ & 60 \\
    Normalization in \eqref{eq:statech} & $E_0$ & 12000 \\
    Rewards and constants in \eqref{eq:RewardCh} & $K_{\rm fin}$,$K_{\rm e}$, $K_{\rm d}$ & 200, 1, 0.1  \\
    Out-of-energy penalty & $K_{\rm en}$ & 100  \\
 \hline
\end{tabular}
\label{Tab:NormCt} 
\end{table}

For starters, let us 
validate the cell-free asymptotic derivations, as they are leveraged in \eqref{opt:ProblemTrack}. 
To that end, Fig.~\ref{fig:SE_Surf_Fire} plots the average user spectral efficiency for different $M$ and $L$; as one would expect, a smaller $M/L$, i.e., more APs per UAV, yields a better spectral efficiency.
In turn, Fig.~\ref{fig:SE_A_Fire} indicates that Thm.~\ref{prop:MMSE_Fire} provides an excellent approximation even in low-dimensional systems, say $M = 3$ and $L = 6$. 
\begin{figure*}[]
     \centering
     \begin{subfigure}[b]{0.495\textwidth}
         \centering
         \includegraphics[scale=0.6]{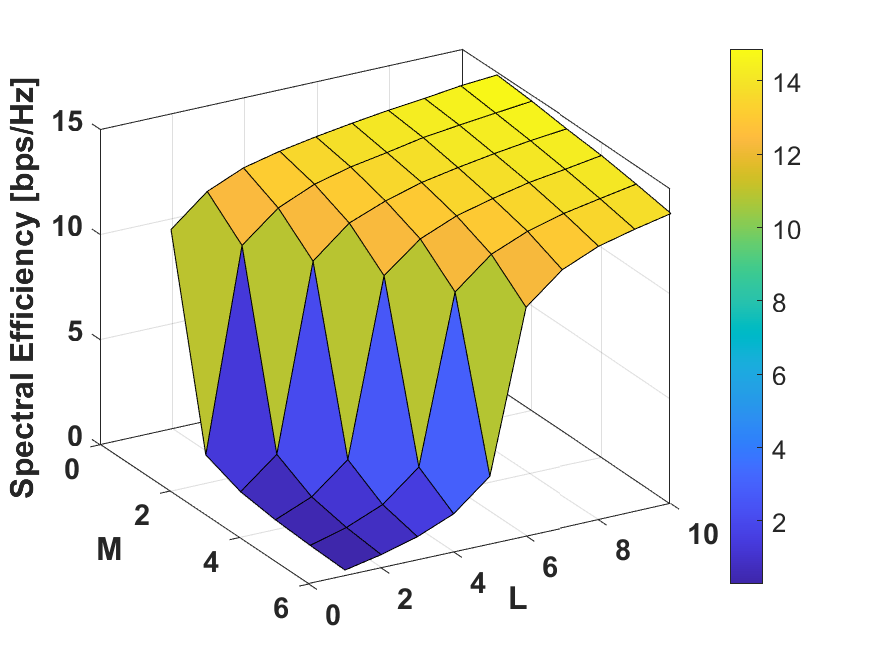}
         \caption{}
         \label{fig:SE_Surf_Fire}
     \end{subfigure}
     \hfill
     \begin{subfigure}[b]{0.495\textwidth}
         \centering
         \includegraphics[scale=0.6]{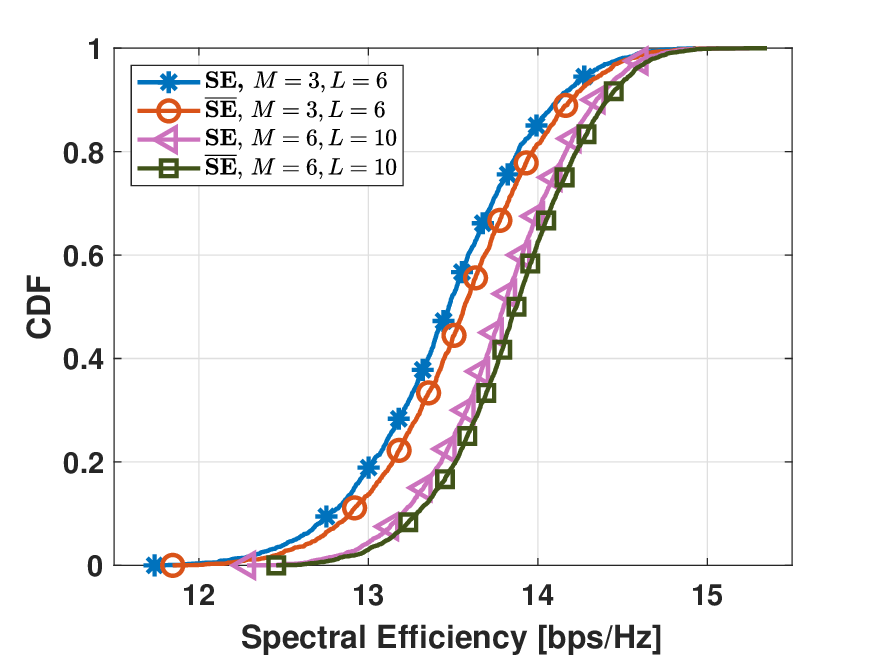}
         \caption{}
         \label{fig:SE_A_Fire}
     \end{subfigure}
     \caption{(a) Spectral efficiency for different $M,L$; (b) validation of Theorem \ref{prop:MMSE_Fire}.
         }\label{Fig:RMT_Fire}
\end{figure*}

Next, the attention shifts to the tracking stage. Fig.~\ref{Fig:Track_Reward} assesses the TD3 algorithm's training performance for $M=1,\dots,4$, where a 3-layer feed-forward neural network is used, each layer containing 256 neurons. In the initial 750 episodes, the replay buffer accumulates 
transitions to ensure that, upon training commencement, every minibatch comprises samples with minimal correlation, enhancing the learning process. Hence, learning truly begins after episode 750, leading to an overall increase in the average reward for all models. Once 3000 episodes are reached, rewards stabilize. For subsequent evaluations, the models saved at the 6000-episode mark are utilized.
Fig.~\ref{fig:Track_M_All} evaluates the trained models over a single realization of a wildfire with
the cost function in \eqref{opt:ProblemTrack}.
Despite the random initialization of the UAV locations, the swarm efficiently repositions itself, ensuring full perimeter coverage by $n=100$. However, as $n$ reaches 300, the perimeter dimensions become too extensive to be managed by $M=1$ UAV, resulting in a significant cost increase. A similar trend is observed for $M=2$ UAVs, but $M=3$ and $M=4$ UAVs manage to maintain a satisfactory performance. Note how, due to the regularization parameter $\Delta$,   the cost function exhibits a high variance. When all points are covered, the function yields a small value. However, if even a single perimeter point remains uncovered, the function experiences sharp spikes.
The coverage metric defined in \eqref{eq:cov} then comes handy, as a smoother complement to the cost function, and hence both values are reported in the sequel.


\begin{figure*}[]
     \centering
     \begin{subfigure}[b]{0.495\textwidth}
         \centering
         \includegraphics[scale=0.6]{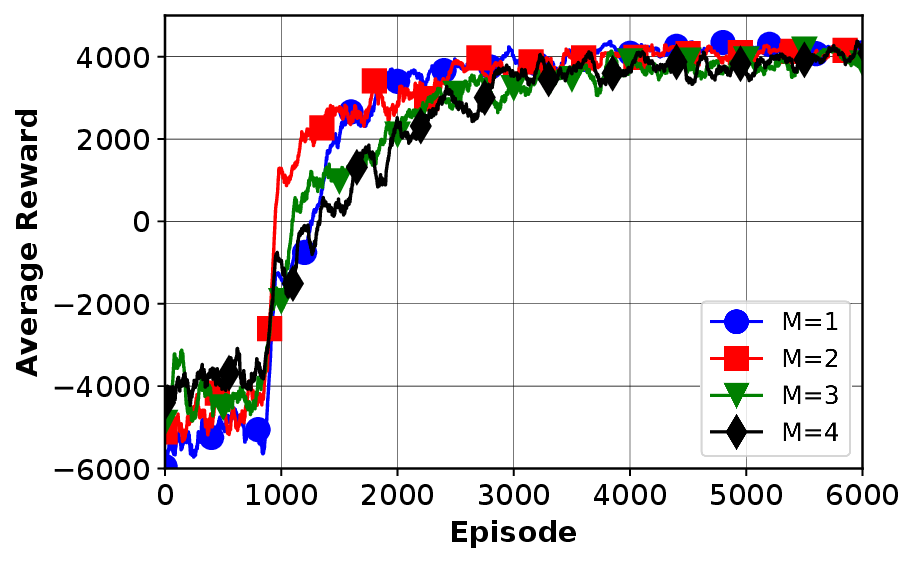}
         \caption{}
         \label{fig:Track_M4}
     \end{subfigure}
     \hfill
     \begin{subfigure}[b]{0.495\textwidth}
         \centering
         \includegraphics[scale=0.6]{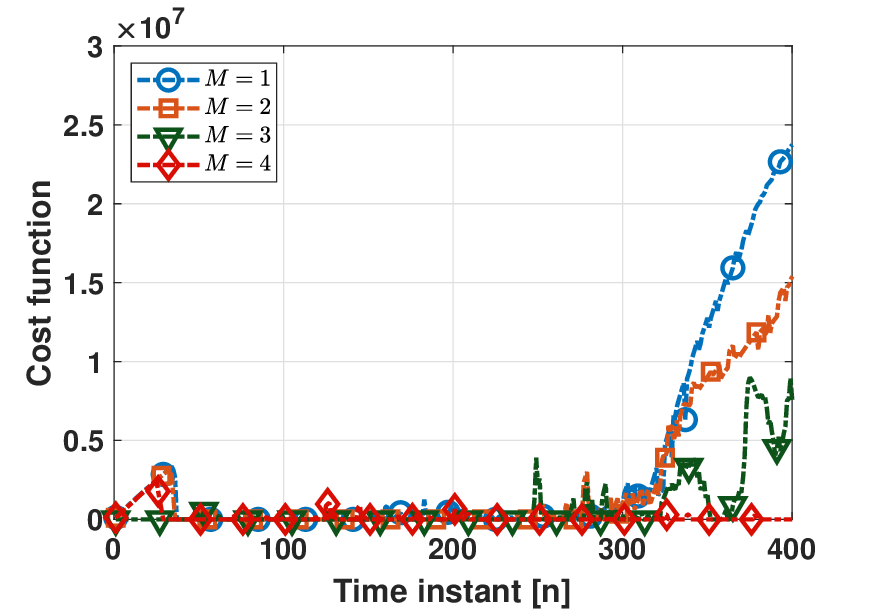}
         \caption{}
         \label{fig:Track_M_All}
     \end{subfigure}
     \caption{For the {tracking} problem (a) average training reward; (b) cost function value as a function of $n$.
         }\label{Fig:Track_Reward}
\end{figure*}

Next, we measure the cumulative distribution function (CDF) of the cost function and fire coverage for $H_{\rm min}=125$ and $H_{\rm max}=150$ at $n=350$ over 1,000 independent wildfire events. 
As one would expect, increasing the number of UAVs results in an increased coverage and a smaller cost function. Particularly, the case with  $M=4$ UAVs achieves a 90\% coverage with very high probability.
Also, note that both figures consider $n=350$, i.e., fires that have already expanded and are large. Although not shown for the sake of conciseness, for $n=100$, meaning for fires that are still small, 90\% coverage is achieved regardless of the number of UAVs.

\begin{figure*}[]
     \centering
     \begin{subfigure}[b]{0.495\textwidth}
         \centering
         \includegraphics[scale=0.6]{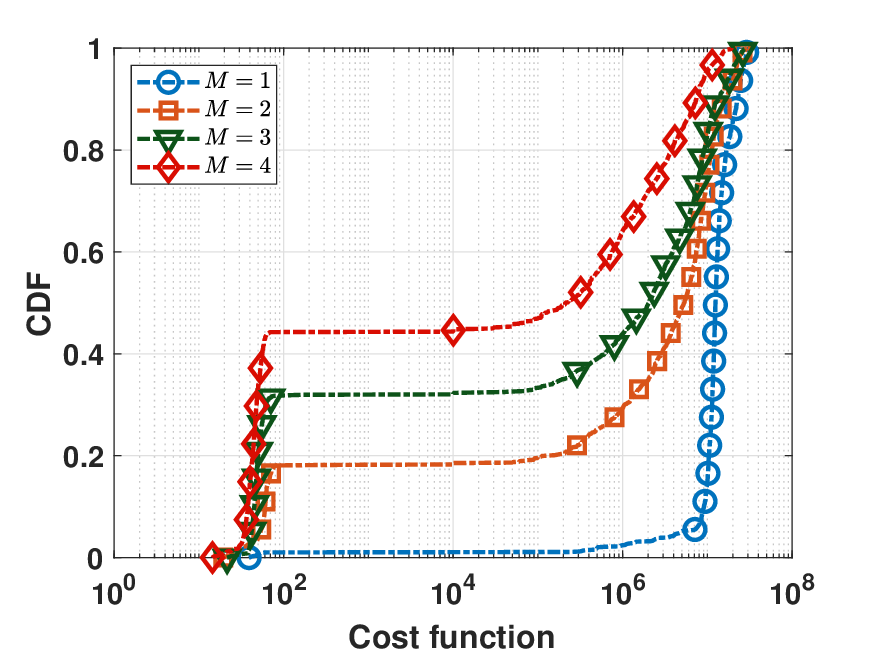}
         \caption{}
         \label{fig:Tr_S1}
     \end{subfigure}
     \hfill
     \begin{subfigure}[b]{0.495\textwidth}
         \centering
         \includegraphics[scale=0.6]{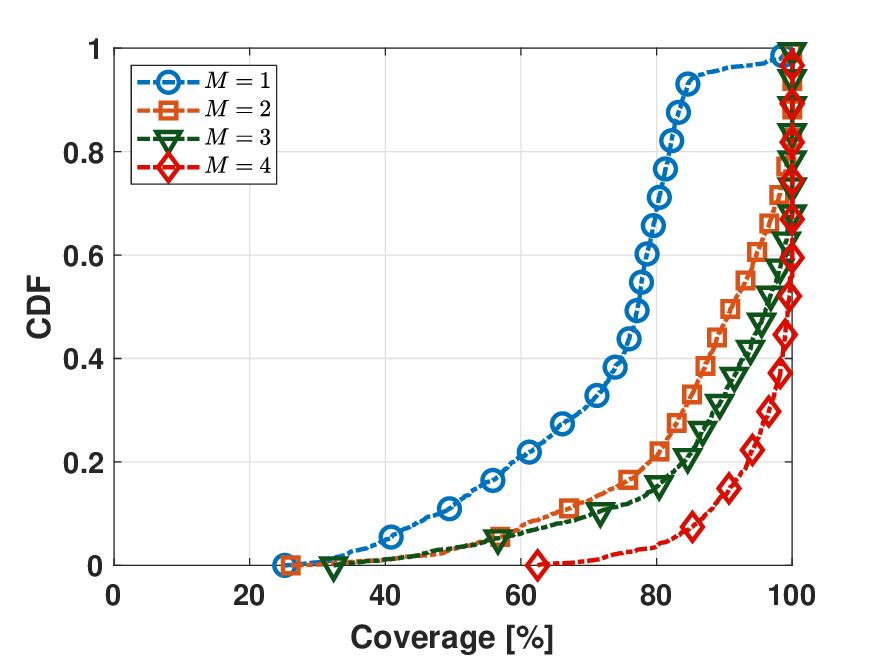}
         \caption{}
         \label{fig:Tr_S2}
     \end{subfigure}
     \caption{At $n=350$ for  $H_{\rm min}=125$ and $H_{\rm max}=150$ (a) cost function; (b) coverage.
         }\label{Fig:Tr_SA}
\end{figure*}

To quantify the impact of the flying altitude on the problem, we include Fig. \ref{Fig:Tr_SB}. These figures measure the swarm's area per pixel (blue) and coverage (red) over 400 time steps and average the results across 1,000 different fire realizations, considering only the tracking stage where UAVs have unlimited energy. Fig. \ref{fig:Tr_S3} corresponds to $H_{\rm min}=125$ and $H_{\rm max}=150$, while Fig. \ref{fig:Tr_S4} is for $H_{\rm min}=100$ and $H_{\rm max}=125$. As expected, reducing the allowed flying altitudes decreases the FoV, leading to a degradation in both the cost function and coverage. This degradation is especially notable when the fire is large, say $n>300$, whereas, at the onset of the event, lower flying altitudes suffice.

\begin{figure*}[]
     \centering
     \begin{subfigure}[b]{0.495\textwidth}
         \centering
         \includegraphics[scale=0.6]{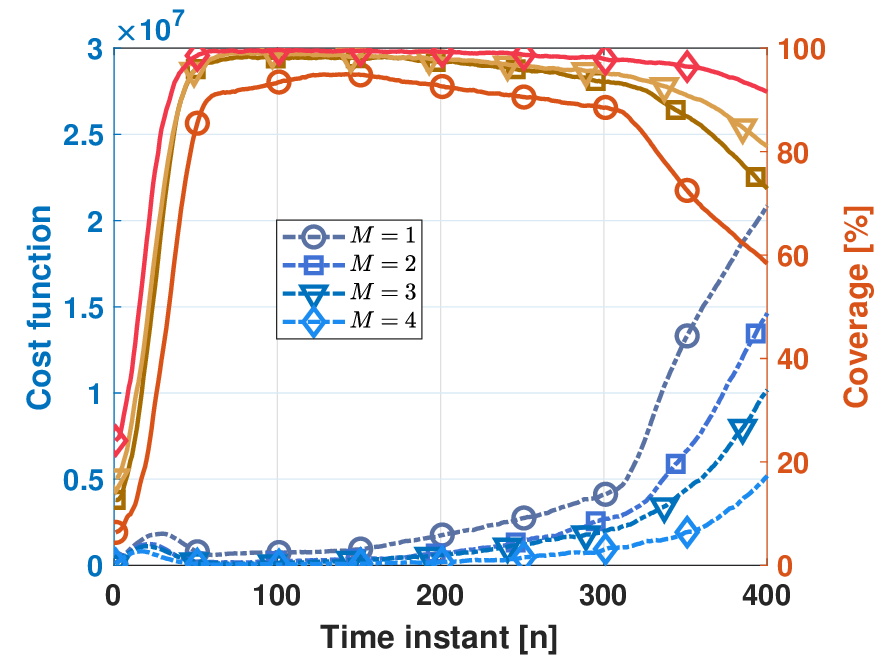}
         \caption{}
         \label{fig:Tr_S3}
     \end{subfigure}
     \hfill
     \begin{subfigure}[b]{0.495\textwidth}
         \centering
         \includegraphics[scale=0.6]{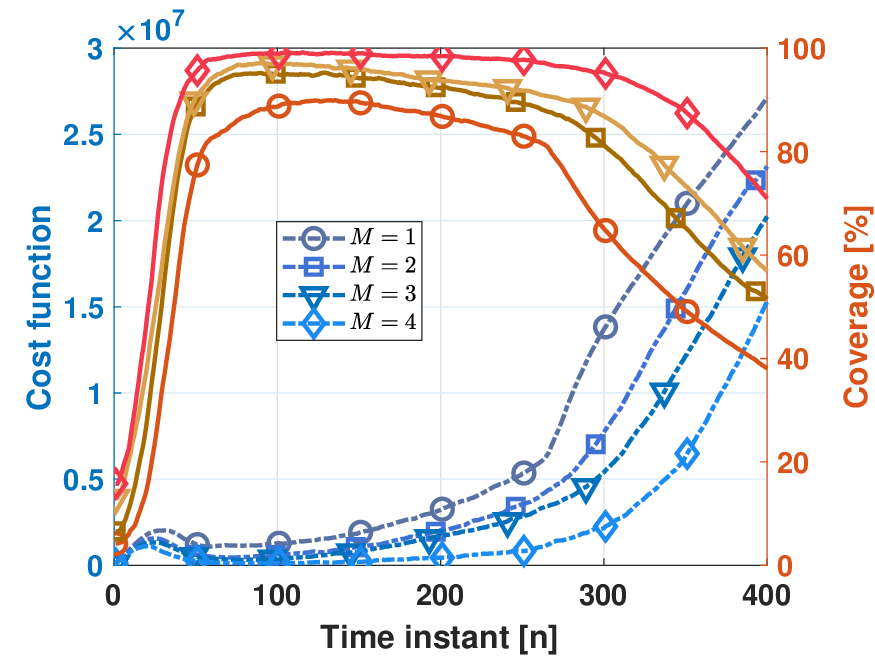}
         \caption{}
         \label{fig:Tr_S4}
     \end{subfigure}
     \caption{Cost and  coverage  for tracking  and $(H_{\rm min},H_{\rm max})$: (a) (125,150) (b) (100,125).
         }\label{Fig:Tr_SB}
\end{figure*}

{
Next, our RL solution is tested against two benchmarks. The first one distributes the UAVs uniformly at random throughout the region and is denoted by U. The second one, denoted by G, places the UAVs around the ignition point following a 2D Gaussian random distribution with covariance $10 \boldsymbol{I}$. Against these benchmarks, two initializations are considered for the RL approach, again uniformly at random or Gaussian, denoted respectively by U+RL and G+RL. Fig.~8 presents results for $M=3$ and $M=4$ with  $H_{\rm min}=125$ and $H_{\rm max}=150$. The following is observed:
    \begin{itemize}
        \item     The U method performs poorly. An approximate value for the coverage achieved by this method can be gauged from the expected coverage with $M$ UAVs, computed as the ratio between the area covered by the UAV's FoV and the total area of the region, namely
        \begin{equation}
        M \frac{ 2^2 \, \mathbb{E} \{ h_m^{(n)} \} ^2 \tan \alpha_1  \tan \alpha_2 }{S^2}.
        \end{equation}
        For  $\mathbb{E} \{h_m^{(n)} \}= (H_{\rm max}+H_{\rm min}) / 2$, the values obtained for $M=3$ and $M=4$ are, respectively, 0.199 and 0.265, consistent with the coverage curves for the U plots. 
        \item The G benchmark performs well during the initial stages. However, as the fire grows large, say $n>150$, both the area per pixel and coverage degrade significantly. In contrast, our method manages to provide solid coverage under both initializations. Starting with the UAVs around the ignition point, corresponding to G+RL, ensures better coverage at first, but the U+RL method rapidly relocates the UAVs to cover the fire effectively too.
    \end{itemize}
}

\begin{figure*}[]
     \centering
     \begin{subfigure}[b]{0.495\textwidth}
         \centering
         \includegraphics[scale=0.6]{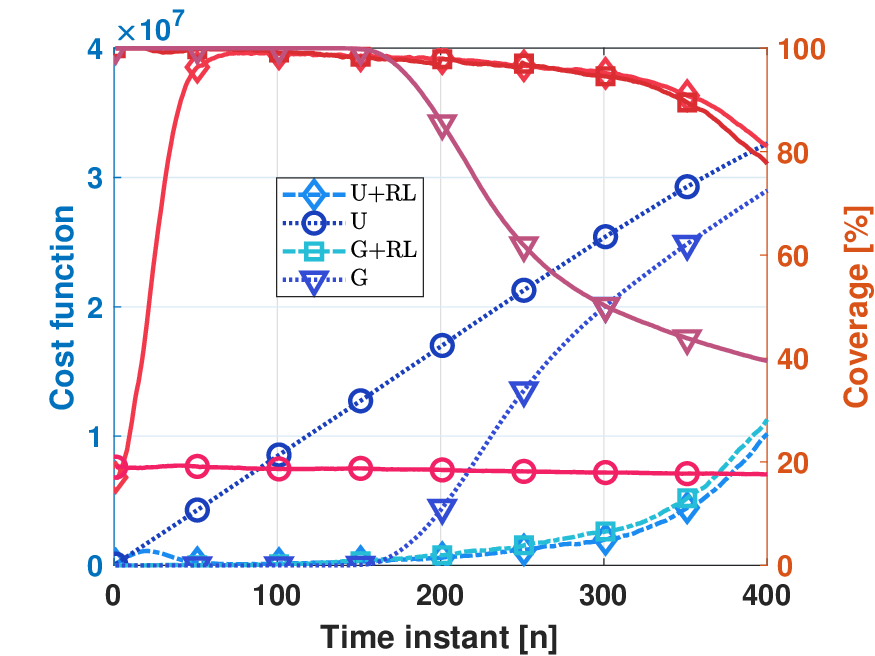}
         \caption{}
         \label{fig:B1}
     \end{subfigure}
     \hfill
     \begin{subfigure}[b]{0.495\textwidth}
         \centering
         \includegraphics[scale=0.6]{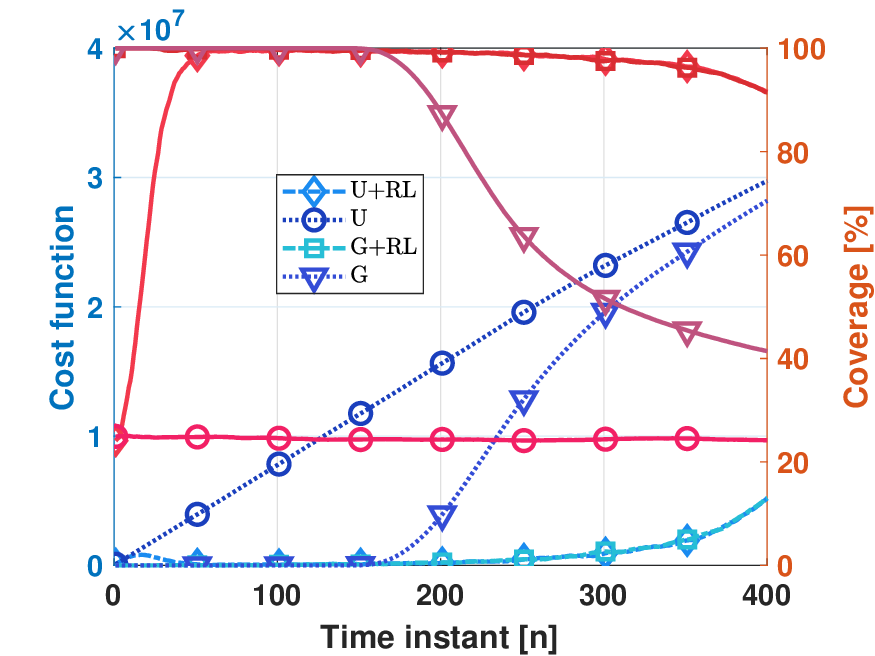}
         \caption{}
         \label{fig:B2}
     \end{subfigure}
     \caption{Cost and  coverage  for tracking under different solutions: (a) $M=3$ (b) $M=4$.
         }\label{Fig:Bench}
\end{figure*}

Before merging the tracking and charging stages, we validate the TD3 solution applied to the charging stage via \eqref{opt:ProblemCharge}, where the same three-layer feedforward neural network is used. For training, episodes are limited to a maximum length of 200. The maximum number of UAVs coincides with that of the tracking stage, i.e., $M=4$. UAV locations are randomly initialized. During  each episode, one UAV intends to reach a charging point, while the remaining UAVs perform random movements. Fig. \ref{fig:Tr_Ch} plots the average training reward versus the episode number for various values of $M$ with  $\gamma=0.85$. After 3,000 episodes, the agent is fully trained, achieving an average reward of around 200---the value of $r_{\rm fin}$---thus resulting in $\boldsymbol{q}_m^{(n)} = \boldsymbol{c}_{c_m}$. Once the agent is trained,  the energy threshold at which the UAV switches from tracking to charging needs to be set. 
For that purpose, a 2D histogram is produced relating the distance between the UAV and the charging point when switching to the charging stage (x-axis), and the amount of energy that the UAV needs to reach the charging point (y-axis). 
Fig. \ref{fig:Eval_Ch} shows such a histogram 
 after evaluating the trained agent over 50,000 episodes. 
Since the UAV should reach the charging point with very high probability, a conservative threshold should be chosen. For example, the UAV switches to charging when, at a certain distance, its level of energy falls below $1.2$ times the maximum energy at that distance in Fig. \ref{fig:Eval_Ch}.
\begin{figure*}[]
     \centering
     \begin{subfigure}[b]{0.495\textwidth}
         \centering
         \includegraphics[scale=0.55]{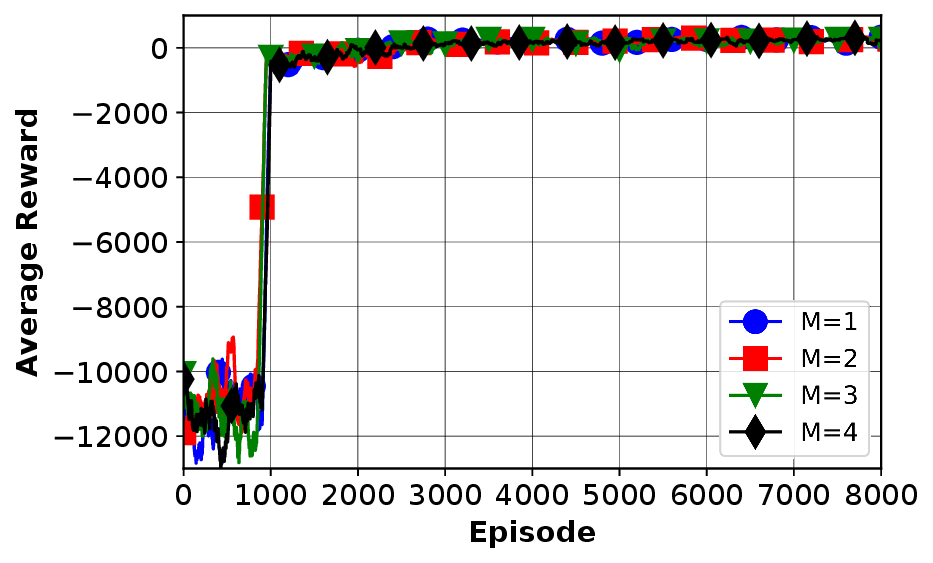}
         \caption{}
         \label{fig:Tr_Ch}
     \end{subfigure}
     \hfill
     \begin{subfigure}[b]{0.495\textwidth}
         \centering
         \includegraphics[scale=0.55]{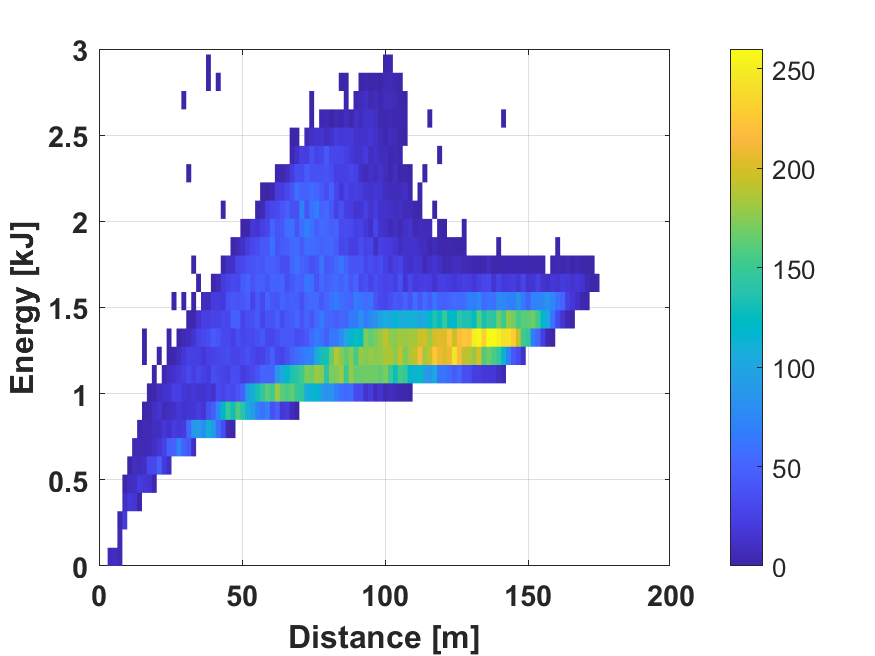}
         \caption{}
         \label{fig:Eval_Ch}
     \end{subfigure}
     \caption{For the {charging} problem (a) average training reward; (b) 2D histogram for the distance vs energy in a trained agent.
         }\label{Fig:Ch_Analysis}
\end{figure*}

Finally, for the combined tracking and charging results, the UAV energies are initialized at random 
from $\mathcal{U}[3.5, E_{\rm max}]$ kJ for different $E_{\rm max}$. This is justified by the UAVs' previous surveillance resulting in different energy values. The minimum energy level for UAVs is set to 3.5 kJ, ensuring that UAVs can reach a charging point in case their initial energy levels are low. We measure the swarm's coverage over 400 time steps and average the results across 1,000 different fire realizations. Fig. \ref{Fig:Cov_1_E} plots the cost function (blue) and coverage (red) given $E_{\rm max}=125$ kJ for various $M$ and $(H_{\rm min},H_{\rm max})$. 
While the initial coverage is small, UAVs rapidly adjust their locations to track the varying perimeter. Although, at each time instant, we aim at minimizing the cost function, the corresponding curves are increasing. This is due to the fire perimeter expanding with time, with a higher cost even if full coverage is achieved. Also, note that the values attained in Fig.~\ref{Fig:Cov_1_E} are similar to those in Fig.~\ref{Fig:Tr_SB}, where infinite energy is assumed, only with a 3-5\% degradation.

To further assess the effects of finite energy batteries, Fig.~\ref{Fig:Cov_2_E} presents results for lower values of $E_{\rm max}$, precisely for $E_{\rm max} = 125$ kJ, $E_{\rm max} = 100$ kJ and $E_{\rm max} = 80$ kJ, for a variety of $M$ with $H_{\rm min}=125$ and $H_{\rm max}=150$. Interestingly, since the UAVs switch to charging more often when their energy levels are lower, the performance worsens. There is a 10\% degradation in coverage for $M=3$  while, for $M=4$, that degradation is  smaller, especially when the fire is extensive. Conversely, when a fire is in its initial stage, coverages above $90$\% can be achieved at all energy levels.

\begin{figure*}[]
     \centering
     \begin{subfigure}[b]{0.495\textwidth}
         \centering
         \includegraphics[scale=0.6]{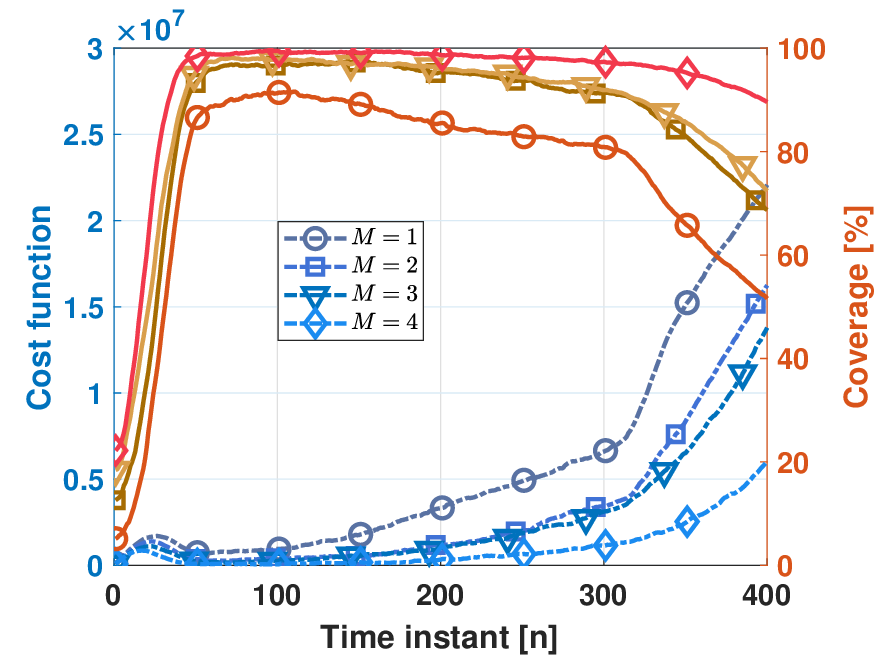}
         \caption{}
         \label{fig:Cov_1_E1}
     \end{subfigure}
     \centering
     \begin{subfigure}[b]{0.495\textwidth}
         \centering
         \includegraphics[scale=0.6]{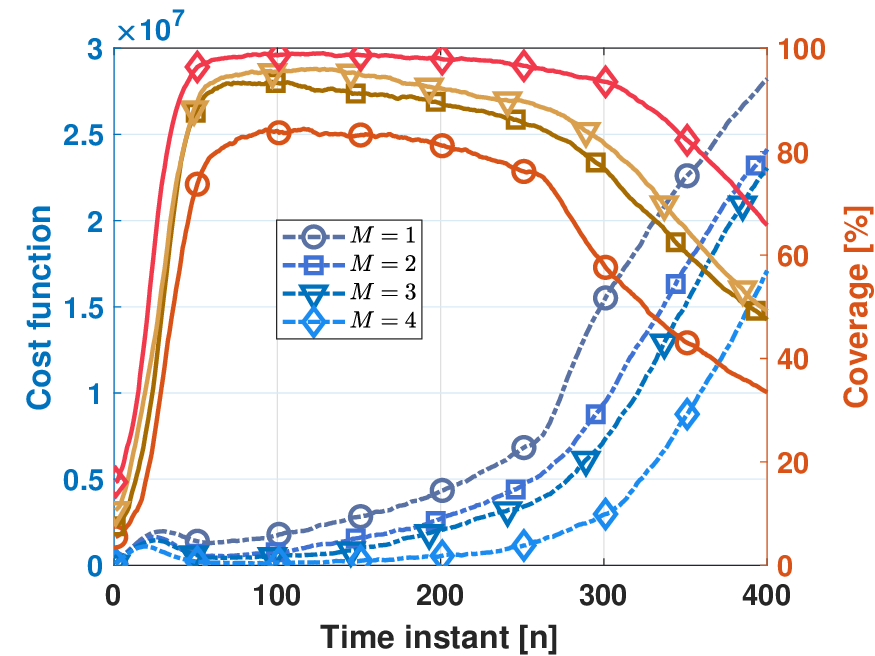}
         \caption{}
         \label{fig:Cov_1_E3}
     \end{subfigure}
     \hfill
     \caption{Cost and  coverage  for $E_{\rm max} = 125$ kJ and $(H_{\rm min},H_{\rm max})$: (a) (125,150) (b) (100,125).
         }\label{Fig:Cov_1_E}
\end{figure*}

\begin{figure*}[]
     \centering
     \begin{subfigure}[b]{0.49\textwidth}
         \centering
         \includegraphics[scale=0.6]{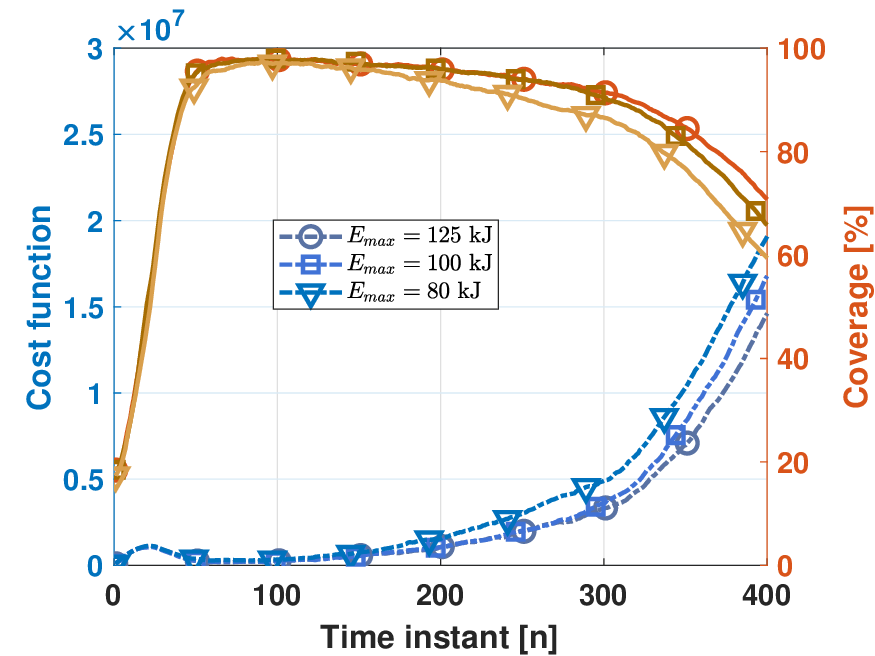}
         \caption{}
         \label{fig:Cov_2_E1}
     \end{subfigure}
     \centering
     \begin{subfigure}[b]{0.499\textwidth}
         \centering
         \includegraphics[scale=0.6]{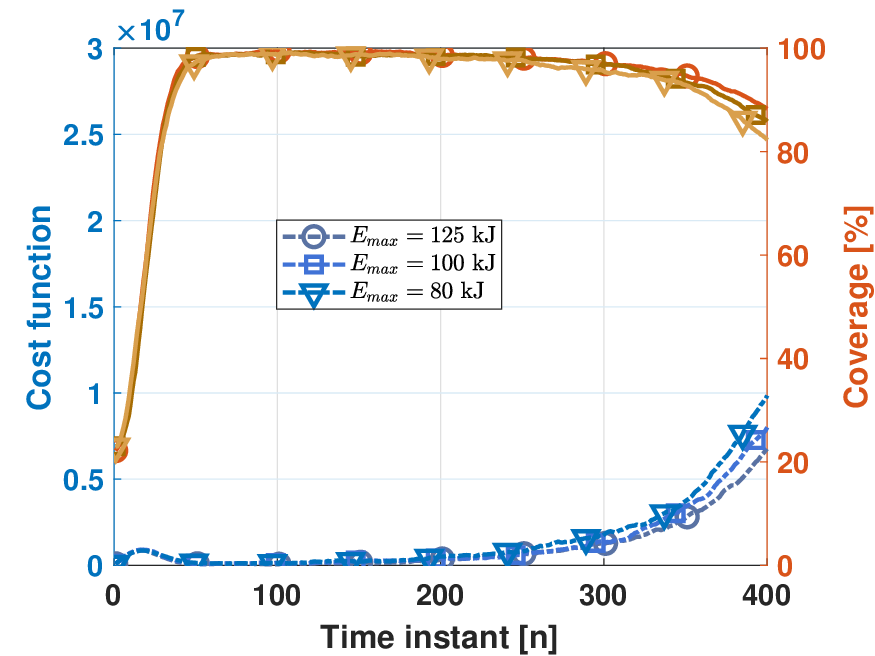}
         \caption{}
         \label{fig:Cov_2_E2}
     \end{subfigure}
     \hfill
     \caption{Cost and coverage for different $E_{\rm max}$ (a) $M=3$, (b) $M=4$. 
         }\label{Fig:Cov_2_E}
\end{figure*}

{
Fig. \ref{Fig:Traj} illustrates the system's operation for $M=3$ UAVs and a given wildfire realization. Precisely, Figs.~\ref{fig:Traj1} and \ref{fig:Traj2} depict the UAV trajectories until $n=150$ and $n=400$, respectively, alongside the wildfire perimeter represented by means of a 2D histogram, where the value of each point represents its importance. For that same realization, the cost function and coverage are depicted in Fig.~\ref{fig:Traj3} whereas Fig.~\ref{fig:Traj4} shows the excess bit rate percentage in the cell-free connectivity constraint given by \eqref{ct:Rate}. At the beginning of the mission, UAV-1 (red) needs to recharge its battery and therefore heads to the closest charging point with coordinates ($225,80,10$). In parallel, UAV-3 (blue) approaches the wildfire and, by $n=36$, full coverage is achieved; see Fig.~\ref{fig:Traj3}. However, the energy levels of UAV-3 are low, and at $n=74$ it switches to charging mode and heads to the charging point located at ($230,225,10$), causing a reduction in coverage. At that time, UAV-1 (red) and UAV-2 (green) have enough energy and are in tracking mode, repositioning themselves closer to the wildfire. At $n=84$, coverage begins to increase, reaching full coverage by $n=140$. Afterwards, the three UAVs continue adjusting their locations to cover the fire perimeter, maintaining the coverage well above $80$\% for the remainder of the mission. Finally, Fig.~\ref{fig:Traj4} presents the excess bit rate in \eqref{ct:Rate} as a percentage of $\frac{B}{N}$. (A value of 0 is assigned during the charging periods as UAVs do not transmit.). The combination of trajectory and power optimization yields a positive gap, indicating that the images are correctly relayed to the network.
}

\begin{figure*}[]
     \centering
     \begin{subfigure}[b]{0.495\textwidth}
         \centering
         \includegraphics[scale=0.63]{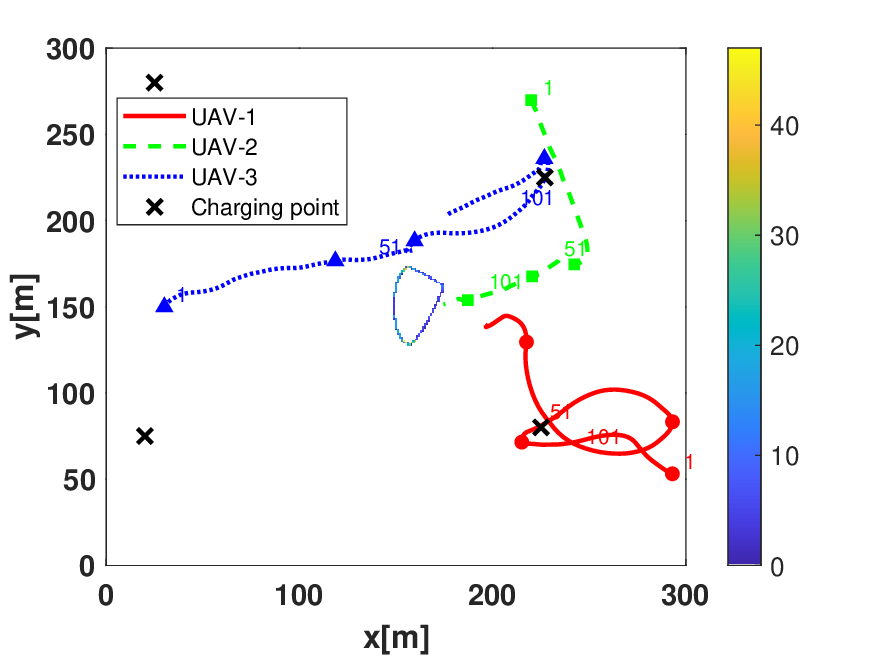}
         \caption{}
         \label{fig:Traj1}
     \end{subfigure}
     \centering
     \begin{subfigure}[b]{0.495\textwidth}
         \centering
         \includegraphics[scale=0.63]{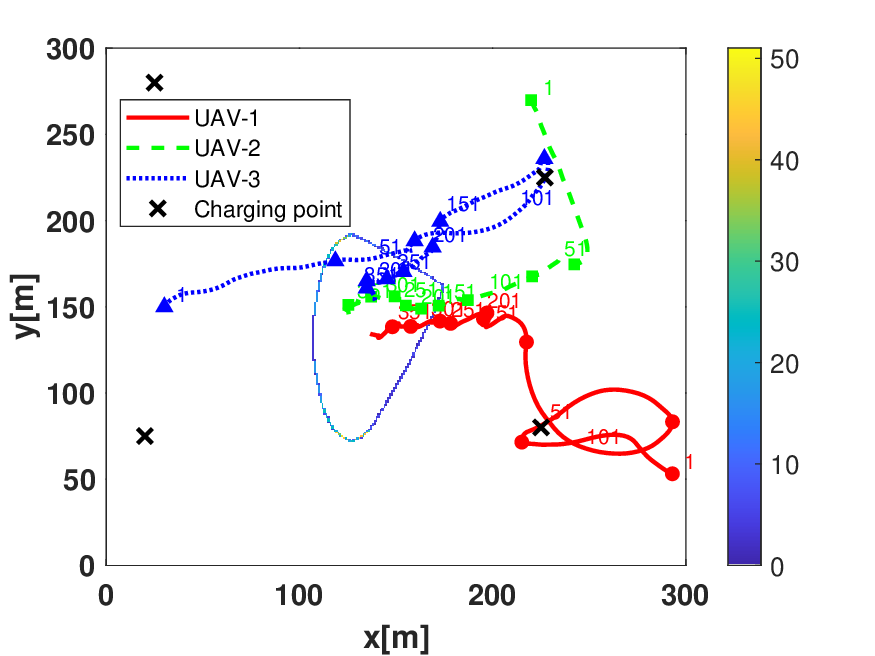}
         \caption{}
         \label{fig:Traj2}
     \end{subfigure}
     \hfill
    \vskip\baselineskip
    \centering
     \begin{subfigure}[b]{0.495\textwidth}
         \centering
         \includegraphics[scale=0.6]{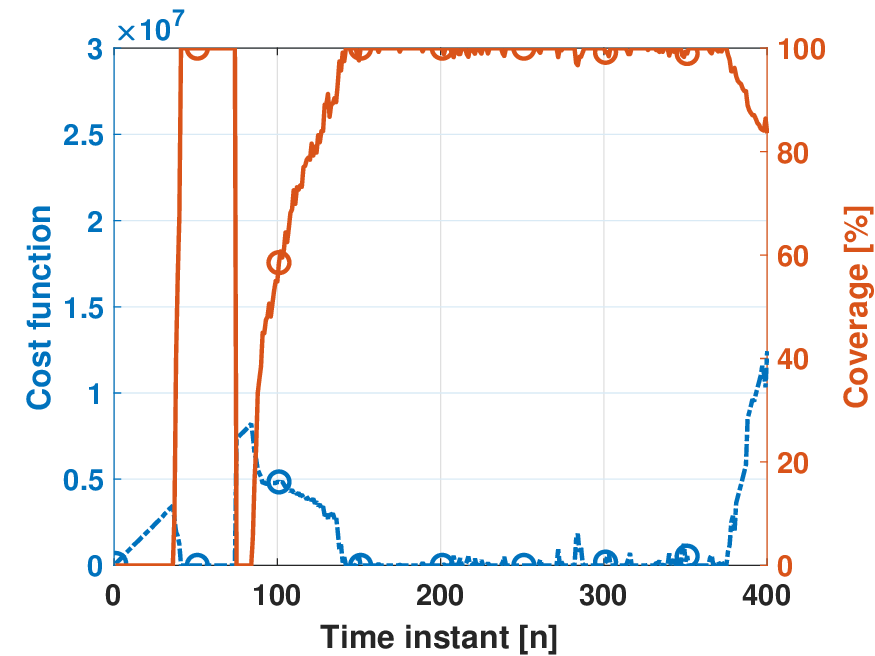}
         \caption{}
         \label{fig:Traj3}
     \end{subfigure}
     \centering
     \begin{subfigure}[b]{0.495\textwidth}
         \centering
         \includegraphics[scale=0.6]{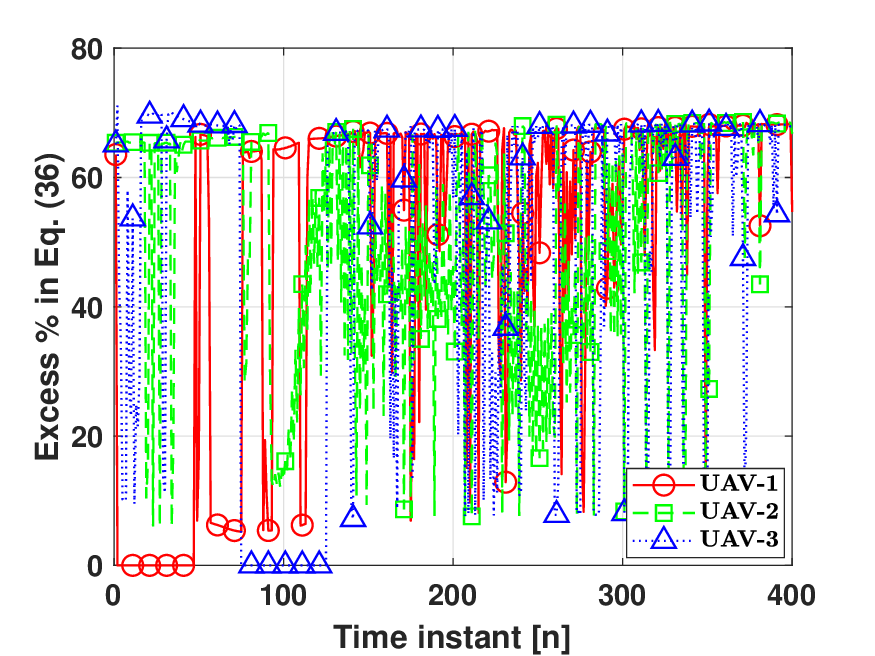}
         \caption{}
         \label{fig:Traj4}
     \end{subfigure}
     \caption{For a given wildfire realization: (a) wildfire and UAV trajectories until $n=150$; (b) wildfire and UAV trajectories until $n=400$; (c) cost function and coverage over the entire mission; (d) excess bit rate in \eqref{ct:Rate}.
         }\label{Fig:Traj}
\end{figure*}

\section{Conclusion}\label{Sec:Concl_Fire}

This paper has considered a cell-free UAV network whose aim is to track a wildfire while satisfying a set of mechanical, energetic, and communication constraints.
Two complex nonconvex optimization problems have been formulated, for tracking and charging, and a reinforcement learning framework has been applied to tackle them. Particularly, the TD3 algorithm has been used. Extensive results have shown that a small swarm of UAVs can reliably provide coverage. Concretely, if the energy levels   and the flying altitudes are moderately high, an average coverage of more than 90\% can be achieved with only a few UAVs, with that coverage shrinking with
the charging level and the altitude. Altogether, the tradeoff among the number of UAVs, energy, and flying altitude, has been established.

\appendices

\section{}
\begin{theorem}\label{th:RMT1}
(\cite[Thm. 1]{6172680}) Let $\boldsymbol{{D}} \in \mathbb{C}^{M \times M}$ and $\boldsymbol{{S}} \in \mathbb{C}^{M \times M}$ be Hermitian nonnegative-definite while  $\boldsymbol{{H}} \in \mathbb{C}^{M \times K}$ is a random matrix with zero-mean independent  vectors $\boldsymbol{h}_k$, each with covariance  $\frac{1}{M}\boldsymbol{R}_k$. In turn, $\boldsymbol{{D}}$  and $\boldsymbol{{R}}_k$ have uniformly bounded spectral norm w.r.t. $M$. For $z>0$ and $M,K \to \infty$,
\begin{equation*}
    \frac{1}{M} \, \mathrm{tr} \! \left[ \boldsymbol{{D}} \big( \boldsymbol{{H}}\boldsymbol{{H}}^{*} + \boldsymbol{{S}} + z\boldsymbol{{I}}_M)^{-1} \right] - \frac{1}{M} \, \mathrm{tr}[ \boldsymbol{{D}} \Tmat] \stackrel{\text{a.s.}}{\to} 0 ,
\end{equation*}
where 
\begin{equation}\label{eq:Tmat}
    \Tmat = \bigg( \frac{1}{M} \sum \limits_{j=1}^K \frac{\boldsymbol{{R}}_j}{1+e_{j}}  + \boldsymbol{{S}} + z\boldsymbol{{I}}_M \bigg)^{\!-1}
\end{equation}
with coefficients $e_k = \text{lim}_{n\xrightarrow{}\infty} e_k^{(n)}$ for
\begin{equation}
    e_k^{(n)} =  \frac{1}{M} \, \mathrm{tr} \! \left[ \boldsymbol{{R}}_k \bigg( \frac{1}{M} \sum \limits_{j=1}^K \frac{\boldsymbol{{R}}_j}{1+e_{j}^{(n-1)}}  + \boldsymbol{{S}} + z\boldsymbol{{I}}_M \bigg)^{\!-1}, \right]
\end{equation}
with initial values $e_k^{(0)}=M$.
\end{theorem}


\section{}\label{proof:SINR_Fire}
Let us drop the time index and define the matrix
\begin{equation}
    \boldsymbol{\Omega}_m = \left(   \Gmh \Pmat{}  \Gmh ^* -  \boldsymbol{\hat{g}}_m \boldsymbol{\hat{g}}_m^*p_m + \boldsymbol{\Sigma} \right)^{\!-1},
\end{equation} 
with $\boldsymbol{P} = \mathrm{diag}\{ p_m \ls \forall  \ls m \}$ and $\boldsymbol{\Omega}_m' = L  \boldsymbol{\Omega}_m$. Then, \eqref{eq:SINR} can be written as
\begin{align}
    \mathrm{SINR}_m & =  \boldsymbol{\hat{g}}_m ^* \boldsymbol{\Omega}_m\boldsymbol{\hat{g}}_m   \, p_m \\
    & = \frac{p_m}{L} \, \mathrm{tr} \! \left[   \boldsymbol{\hat{g}}_m\boldsymbol{\hat{g}}_m^* \boldsymbol{\Omega}_m' \right].
\end{align}
For $M$,$L$ $\xrightarrow{}\infty$, using \cite[Lemma 4]{6172680} and Theorem \ref{th:RMT1},  
\begin{equation}\label{eq:LargeMKSINR}
    \frac{p_m}{L} \, \mathrm{tr} \! \left[   \boldsymbol{\hat{g}}_m\boldsymbol{\hat{g}}_m^* \boldsymbol{\Omega}_m' \right] -  \frac{p_m}{L} \mathrm{tr} [ \boldsymbol{\Gamma}_m \Tmat_m ] \stackrel{\text{a.s.}}{\to} 0.
\end{equation}
In our case, the role of $\big( \boldsymbol{{H}}\boldsymbol{{H}}^{*} + \boldsymbol{{S}} + z\boldsymbol{{I}}_M)^{\!-1}$ in Theorem \ref{th:RMT1} is played by $\boldsymbol{\Omega}_m'$. There is a direct mapping between the terms in the aforementioned theorem and our problem, namely (\emph{i}) $\boldsymbol{{D}}= \boldsymbol{\Gamma}_m \, p_m$, (\emph{ii}) $\boldsymbol{{R}}_j= \boldsymbol{\Gamma}_j \, p_j$, and (\emph{iii}) $\boldsymbol{{S}} + z\boldsymbol{{I}}_M=\frac{1}{L} \boldsymbol{\Sigma}$ with $\Tmat_m$ following the structure of $\Tmat$ in Theorem \ref{th:RMT1}, 
\begin{equation}
    \Tmat_m = \bigg( \frac{1}{L} \sum \limits_{j \neq m }^M \frac{\boldsymbol{\Gamma}_j}{1+e_{j}} \, p_j  + \frac{1}{L} \boldsymbol{\Sigma} \bigg)^{\!-1}.
\end{equation}
The  coefficients can be calculated as $e_{j} = \lim_{t \to \infty} e_{j}^{(t)} $ with
\begin{align}
    e_{j}^{(t)} & = p_j \, \mathrm{tr} \Bigg[ \boldsymbol{\Gamma}_j \bigg(  \sum \limits_{i \neq j }^M \frac{\boldsymbol{\Gamma}_i}{1+e_{i}^{(t-1)}} \, p_i  +  \boldsymbol{\Sigma} \bigg)^{\!-1}  \Bigg] .
\end{align}
The fixed-point algorithm can be used to compute $e_{j}$ and
has been proved to converge \cite{6172680}. Finally, since matrices $\boldsymbol{\Gamma}_m$ and $\Tmat_m$ are diagonal, \eqref{eq:LargeMKSINR} can be written as
\begin{align}
    \frac{p_m}{L} \mathrm{tr} [ \boldsymbol{\Gamma}_m \Tmat_m ]  = {p_m}\, \mathrm{tr} \! \left[  \boldsymbol{\Gamma}_m \bigg(  \sum \limits_{i \neq m }^M \frac{\boldsymbol{\Gamma}_i}{1+e_{i}} \, p_i  +  \boldsymbol{\Sigma} \bigg)^{\!-1} \right],
\end{align}
and, after some algebra, Theorem \ref{prop:MMSE_Fire} is proved.

\bibliography{references}
\bibliographystyle{ieeetr}
\end{document}